\documentclass{llncs}
\usepackage[utf8]{inputenc}
\usepackage[colorlinks=true, citecolor=blue, linkcolor=black]{hyperref}
\usepackage{listings}
\usepackage{caption}
\usepackage{amsmath,amssymb}
\usepackage{tikz}
\usetikzlibrary{shapes,shapes.geometric,arrows,fit,calc,positioning,automata,decorations.pathmorphing}
\usepackage{array}
\usepackage{marginnote}
\usepackage{xcolor}
\usepackage{bm}
\usepackage{xspace}
\usetikzlibrary{snakes}
\usepackage{enumerate}

\def\DRAFT{0}   

\newcommand{\pssi}{\par\smallskip\indent}
\newcommand{\pssn}{\par\smallskip\noindent}
\newcommand{\pmsi}{\par\medskip\indent}
\newcommand{\pmsn}{\par\medskip\noindent}

\newcommand{\pbsn}{\par\bigskip\noindent}
\newcommand{\pnsi}{\par\indent}
\newcommand{\pnsn}{\par\noindent}
\newcommand{\mylabel}[1]{\label{#1}\if\DRAFT1\marginnote{\footnotesize\color{red}#1}\fi}
\newcommand{\prlabel}[1]{\if\DRAFT1\marginnote{\footnotesize\color{red}#1}\fi}

\newcommand{\emdef}[1]{\textsf{\emph{#1}}}

\newcommand{\PD}{{\mathsf{PD}}}
\DeclareMathOperator{\APD}{PD}

\newcommand{\lang}{\mathcal{L}}
\newcommand{\apd}{{\hat a}_{\APD}}

\newcommand{\quo}[2]{{#1}^{-1} {#2}}

\newcommand{\eqmod}{\doteq}
\DeclareMathOperator{\lf}{\textsf{n}}

\newcommand{\regexp}{\textbf{r}}

\newcommand{\regexpc}{\textbf{v}}
\newcommand{\letter}{x}
\newcommand{\lettera}{y}
\DeclareMathOperator{\REspec}{\reg \sspec[\Sigma]}
\DeclareMathOperator{\const}{\textsf{c}}
\DeclareMathOperator{\SoS}{\textsf{SS}}
\DeclareMathOperator{\ws}{\textsf{s}}

\newcommand{\concat}{\cdot}
\newcommand{\disj}{+}
\newcommand{\speca}{F}
\newcommand{\specb}{G}
\newcommand{\specc}{C}
\newcommand{\specd}{A}

\newcommand{\tuple}[2]{#1/#2}

\newcommand{\pair}[2]{\left( #1, #2 \right)}

\newcommand{\sseq}[0]{\subseteq}
\newcommand{\sequ}[1]{\langle#1\rangle}  

\newcommand{\pto}[0]{\dashrightarrow} 
\newcommand{\undef}[0]{\bot}
\newcommand{\lop}[0]{\odot}

\newcommand{\tcomp}[0]{\circledcirc} 

\DeclareMathOperator{\REG}{REG}

\DeclareMathOperator{\sspec}{SSP}
\DeclareMathOperator{\srel}{PSP}

\DeclareMathOperator{\editops}{ED}
\DeclareMathOperator{\srelinvar}{PSP^{invar}}
\DeclareMathOperator{\wo}{wo}
\DeclareMathOperator{\al}{alph}

\DeclareMathOperator{\spleft}{left}
\DeclareMathOperator{\spright}{right}
\DeclareMathOperator{\sprset}{rset}

\DeclareMathOperator{\Labels}{Labels}

\DeclareMathOperator{\aarr}{ARR}    
\DeclareMathOperator{\outof}{notIn}
\DeclareMathOperator{\mon}{mon}

\newcommand{\garr}{\aarr_\Gamma}    

\newcommand{\None}{\ensuremath{\mathtt{None}}\xspace}
\newcommand{\False}{\ensuremath{\mathtt{False}}\xspace}
\newcommand{\nonEmptyW}[0]{\ensuremath{\mathtt{nonEmptyW}}}
\newcommand{\emptyP}[0]{\ensuremath{\mathtt{emptyP}}}
\newcommand{\identityP}[0]{\ensuremath{\mathtt{identityP}}}
\newcommand{\functionalityP}[0]{\ensuremath{\mathtt{functionalityP}}}

\newcommand{\cI}{\ensuremath{\mathcal{I}}}
\newcommand{\cali}{\ensuremath{\mathcal{I}}}
\newcommand{\calr}{\ensuremath{\mathcal{R}}}
\newcommand{\call}{\ensuremath{\mathcal{L}}}
\newcommand{\N}[0]{\mathbb N}

\newcommand{\srep}[1]{\underline{#1}} 

\newcommand{\eset}{\emptyset}    
\newcommand{\ew}{\varepsilon}    
\newcommand{\ews}{\bm e}         
\newcommand{\bmoplus}{\bm\oplus}

\newcommand{\erel}{\bm\oslash}    
\newcommand{\sany}{\bm\forall}
\newcommand{\sone}{\bm\exists}
\newcommand{\snone}{{/\!\!\!\sone}}
\newcommand{\sdiff}[2]{#1/#2{\bm\neq}}
\newcommand{\ssame}[1]{#1/{\bm=}}

\newcommand{\spp}{\mathsf{p}}
\newcommand{\sppinv}{\spp^{-1}}

\newcommand{\rinp}[0]{\downarrow} 
\newcommand{\rout}[0]{\uparrow} 
\newcommand{\nel}[1]{\ew_{\mon#1}} 

\newcommand{\lof}[1]{\mathcal{L}(#1)}   
\newcommand{\rel}[1]{\mathcal{R}(#1)}   
\newcommand{\beh}[1]{\mathcal{I}(#1)}

\newcommand{\card}[1]{|#1|}
\newcommand{\sz}[1]{|#1|}
\newcommand{\szabc}[1]{|#1|}
\newcommand{\szw}[1]{\|#1\|}    

\newcommand{\reg}[1]{\REG #1}
\newcommand{\rer}[0]{\mathbf r}
\newcommand{\res}[0]{\mathbf s}

\newcommand{\auta}[0]{\hat a}
\newcommand{\autb}[0]{\hat b}
\newcommand{\autg}[0]{\hat g}
\newcommand{\auth}[0]{\hat h}
\newcommand{\tri}[0]{\hat i}
\newcommand{\trs}[0]{\hat s}
\newcommand{\trt}[0]{\hat t}
\newcommand{\trsinv}[0]{\hat s^{-1}}
\newcommand{\trtinv}[0]{\hat t^{-1}}
\newcommand{\tru}[0]{\hat u}
\newcommand{\trv}[0]{\hat v}

\begin{document}

\title{Regular Expressions and Transducers over Alphabet-invariant and User-defined Labels\thanks{Research supported by NSERC (Canada) and by FCT project UID/MAT/00144/2013 (Portugal).}}

\author{Stavros Konstantinidis\inst{1} \and Nelma Moreira\inst{2} \and Rogerio Reis\inst{2} \and Joshua Young\inst{1}}

\institute{
Saint Mary's University, Halifax, Nova Scotia, Canada,\\
\email{s.konstantinidis@smu.ca}, 
\email{jyo04@hotmail.com}
\and
CMUP \& DCC, Faculdade de Ci{\^e}ncias da Universidade do Porto,
Rua do Campo Alegre, 4169-007 Porto Portugal
\email{\{nam,rvr\}@dcc.fc.up.pt}}

\maketitle

\begin{abstract}
We are interested in regular expressions and transducers that represent word relations in an alphabet-invariant way---for example, the set of all word pairs u,v where v is a prefix of u independently of what the alphabet is. Current software systems of formal language objects do not have a mechanism to define such objects. We define transducers in which transition labels involve what we call set specifications, some of which are alphabet invariant. In fact, we give a more broad definition of automata-type objects, called labelled graphs, where each transition label can be any string, as long as that string represents a subset of a certain monoid. Then, the behaviour of the labelled graph is a subset of that monoid. We do the same for regular expressions. We obtain extensions of  a few classic algorithmic constructions on ordinary regular expressions and transducers at the broad level of labelled graphs and in such a way that the computational efficiency of the extended constructions is not sacrificed. For regular expressions with set specs we obtain the corresponding partial derivative automata. For transducers with set specs we obtain further algorithms that can be applied to questions about independent regular languages, in particular the witness version of the independent property satisfaction question. 
\keywords{Alphabet-invariant transducers, regular expressions, partial derivatives, algorithms, monoids}
\end{abstract}

\section{Introduction}\mylabel{SEC:Introduction}
We are interested in 2D regular expressions and transducers over alphabets whose cardinality is not fixed, or whose alphabet is even unknown. In particular, assume that the alphabet is 
\[\Gamma=\{0,1,\ldots,n-1\}\] 
and consider the 2D regular expression
\[ \big(0/0+\cdots+(n-1)/(n-1)\big)^*\big(0/\ews+\cdots+(n-1)/\ews\big)^*, \]
where $\ews$ is the symbol for the empty string. This 2D regular expression has $O(n)$ symbols and describes the prefix relation, that is, all word pairs $(u,v)$ such that $v$ is a prefix of $u$. Similarly, consider the transducer in Fig.~\ref{FIG:sym:motiv}, which has $O(n^2)$ transitions. Current software systems of formal language objects require users to enter all these transitions in order to define and process the transducer. We want to be able to use special labels in transducers such as those in the transducer $\trt_{\rm sub2}$ in Fig.~\ref{FIG:sym:transd}. In that figure, the label $(\ssame{\sany})$ represents the set $\{(a,a)\mid a\in\Gamma\}$ and the label $(\sdiff{\sany}{\sany})$ represents the set $\{(a,a')\mid a,a'\in\Gamma,a\neq a'\}$ (these labels are called pairing specs). Moreover that transducer has only a fixed number of 5 transitions. Similarly, using these special labels, the above 2D regular expression can be written as
\[(\ssame{\sany})^*(\sany/\ews)^*.\]
Note that the new regular expression as well as the new transducer in Fig.~\ref{FIG:sym:transd} are \emdef{alphabet invariant} as they contain no symbol of the intended alphabet $\Gamma$---precise definitions are provided in the next sections.

	\begin{figure}\begin{center}
		\begin{tikzpicture}[>=stealth', shorten >=1pt, auto, node distance=4.50cm,initial text={}]
			\tikzset{sstate/.style={state,inner sep=3pt, minimum size=6pt}}
			\node [initial, sstate] (i) {$0$};
			\node [sstate, accepting] (s) [right of=i] {$1$};
			\node [sstate, accepting] (f) [right of=s] {$2$};
			\path [->]
			      (i) edge node [above] {$a/a'$} (s)
			      (i) edge node [below] {$(\forall a,a'\in\Gamma:a\neq a')$} (s)
			      (s) edge node [above] {$a/a'$} (f)
			      (s) edge node [below] {$(\forall a,a'\in\Gamma:a\neq a')$} (f)
			      (i) edge [loop above] node [above] {$a/a\>(\forall a\in\Gamma)$} ()
			      (s) edge [loop above] node [above] {$a/a\>(\forall a\in\Gamma)$} ()
			      (f) edge [loop above] node [above] {$a/a\>(\forall a\in\Gamma)$} ()
			;
		\end{tikzpicture}
		\pmsn\parbox{0.85\textwidth}{\caption{The  transducer realizes the relation  of all $(u,v)$ such that $u\neq v$ and the Hamming distance of $u,v$ is at most 2.}\label{FIG:sym:motiv}\prlabel{FIG:sym:motiv}}
	\end{center}\end{figure}
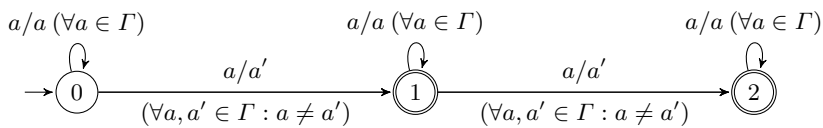

We also want to be able to define algorithms that work \textsl{directly} on regular expressions and transducers with special labels, without of course having to expand these labels to ordinary ones. Thus, for example, we would like to have an efficient algorithm that computes whether a pair $(u,v)$ of words is in the relation realized by the transducer in Fig.~\ref{FIG:sym:transd}, and an efficient algorithm to compute the composition of two transducers with special labels.
\pssi
We start off with the broad concept of a set $B$ of special labels, called \emdef{label set}, where each special label $\beta\in B$ is simply a string that represents a subset $\cali(\beta)$ of a monoid $M$. Then we define type $B$ automata  (called \emph{labelled graphs}) in which every transition label is in $B$. Similarly we consider type $B$ regular expressions whose base objects (again called labels) are elements of $B$ and  represent monoid subsets. Our first set of results apply to any user-defined set $B$ and  monoid $M$. Then, we consider further results specific to the cases of (i) 1D regular expressions and automata (monoid $M=\Gamma^*$), (ii) 2D regular expressions and transducers (monoid $M=\Gamma^*\times\Gamma^*$) with special labels (called \emph{set specs}). A labelled graph in this work can possibly be considered to be a compact version of an automaton\footnote{While the labels of the automaton can possibly represent sets in compressed format, we have no intention to define any specific compression method, as the syntax of the labels is left to the application.}  over the monoid $M$ in the sense of \cite{Sak:2015}. 
\pnsi
We emphasize that we do not attempt to define regular expressions and automata outside of monoids; rather we use  monoid-based regular expressions and automata as a foundation such that (i) one can define such objects with a priori unknown label sets $B$, as long as each of the labels represents a subset of a known monoid; (ii) many known algorithms and constructions on  monoid-based regular expressions and automata are extended to work directly and efficiently on the type $B$ objects. 
\pnsi
We also mention the framework of symbolic automata and transducers of \cite{VHLMB:2012,Vea:2013}. In that framework, a transition label is a logic predicate describing a set of domain elements (characters). The semantics of that framework is very broad and includes the semantics of label sets in this work. As such, the main algorithmic results in \cite{VHLMB:2012,Vea:2013} do not include time complexity estimates.  Moreover, outside of the logic predicates there is no provision to allow for user-defined labels and related algorithms working directly on these labels. 
\pnsi
The role of a label set is similar to that of an alphabet, or of the set of regular expressions: it enables users to represent sets of interest. While some of our results apply to regular expressions and labelled graphs over any user-defined label set, the particular case where the label set is the set of pairing specs allows us to rewrite ordinary transducers, like the one in Fig.~\ref{FIG:sym:motiv}, in a simpler form such that algorithms can work directly on these simpler transducers. In particular, we can employ simple transducers like the one in Fig.~\ref{FIG:sym:transd} to answer the satisfaction question in the theory of independent formal languages. While it seems that pairing specs work well with nondeterministic automata and transducers, this might not be true when dealing with deterministic ones. We discuss this issue further in the last section of the paper.

\if\DRAFT1\pbsn\textbf{\color{blue} Add more refs ?}\pbsn\fi

The paper is organized as follows. The next section makes some assumptions about \emdef{alphabets} $\Gamma$ of non-fixed size. These assumptions are needed in algorithms that process regular expressions and automata with labels involving $\Gamma$-symbols. Section~\ref{SEC:sym:sets} defines the set of \emdef{set specs}, a particular kind of a label set in which each element represents a subset of $\Gamma$ or the empty string, and presents basic algorithms on set specs. Section~\ref{SEC:pairing} defines the set of \emdef{pairing specs}, a particular kind of a label set that is used for transducer-type labelled graphs. Some of these pairing specs are \emdef{alphabet invariant}. Section~\ref{SEC:labelsets} discusses the general concept of a \emdef{label set}, with set specs and  pairing specs being two specific examples of label sets. Each label set $B$ has a behaviour $\cali$ and refers to a monoid, denoted by $\mon B$; that is, $\cali(\beta)$ is a subset of $\mon B$ for any label $\beta\in B$. Section~\ref{SEC:graphs} defines type $B$ labelled graphs $\autg$ and their behaviours $\cali(\autg)$. When $B$ is the set of pairing specs then $\autg$ is a transducer-type graph and realizes a word relation. Section~\ref{SEC:rational} establishes that the rational operations of union, catenation and Kleene star on ordinary automata and transducers work without complications on any labelled graphs. Section~\ref{SEC:sym:RE} defines regular expressions $\rer$ over any label set $B$ and their behaviour $\cali(\rer)$, and establishes the equivalence of type $B$ graphs and type $B$ regular expressions (see Theorem~\ref{TH:thompson} and Corollary~\ref{COR:state:elim}). Section~\ref{SEC:derivatives} considers the concept of linear form of a regular expression over a label set, which leads to the definition of its corresponding partial derivative graph. Then, for regular expressions over set specs 
   it presents the development of the corresponding finite automaton 
    that is equivalent to the regular expression (see Theorem~\ref{TH:apd}). Section~\ref{SEC:product} considers the possibility of defining `higher level' versions of product constructions that  work on automata/transducers over known monoids. To this end, we consider the concept of \emdef{polymorphic operation} `$\odot$' that is partially defined between two elements of some labels sets $B,B'$, returning an element of some label set $C$, and also partially defined on the elements of the monoids $\mon B$ and $\mon B'$, returning an element of the monoid $\mon C$. In this case, if $\odot$ is known to work on automata/transducers over $\mon B,\mon B'$ then it would also work on type $B,B'$ graphs (see Theorem~\ref{TH:product}). Section~\ref{SEC:sym:nfa} presents some basic algorithms on automata with set specs and transducers with set specs. Section~\ref{SEC:compose} defines the composition of two transducers with set specs such that the complexity of this operation is consistent with the case of ordinary transducers (see Theorem~\ref{TH:compose}). Section~\ref{SEC:functionality} considers the questions of whether a transducer with set specs realizes an identity and whether it realizes a function. It is shown that both questions can be answered with a time complexity consistent with that in the case of ordinary transducers (see Theorem~\ref{TH:identity} and Theorem~\ref{TH:funct}). Section~\ref{SEC:independence} shows that, like ordinary transducers, transducers with set specs that define independent language properties can be processed directly (without expanding them) and efficiently to answer the witness version of the property satisfaction question for regular languages (see Corollary~\ref{COR:sat} and Example~\ref{EX:sat}).  Finally, the last section contains a few concluding remarks and directions for future research.

\section{Terminology and Alphabets of Non-fixed Size}\mylabel{SEC:Terminology}
The set of positive integers is denoted by $\N$. Then, $\N_0=\N\cup\{0\}$. Let $S$ be a set. We denote the \emdef{cardinality} of $S$ by $\left| S \right|$ and the set of all subsets of $S$ by $2^{S}$. To indicate that $\phi$ is a \emdef{partial mapping} of a set $S$ into a set $T$ we shall use the notation
\[\phi: S\pto T \]
We shall  write $\phi(s)=\undef$ to indicate that $\phi$ is not defined on $s\in S$.
%
\pssi
An \emdef{alphabet space} $\Omega$ is an infinite and totally ordered set whose elements are called \emdef{symbols}. We shall assume that $\Omega$ is fixed and contains the digits $0,1,\ldots,9$, which are ordered as usual, as well as the \emdef{special symbols}
\[
\bm\sany\quad \bm\sone \quad \bm\snone \quad{\bm=} \quad {\bm\neq}\quad \bm/ \quad \ews \quad \bmoplus \quad \erel
\]  
We shall denote by `$<$' the total order of $\Omega$. As usual we use the term \emdef{string} or \emdef{word} to refer to any finite sequence of symbols. The \emdef{empty string} is denoted by $\ew$. For any string $w$ we say that $w$ is \emdef{sorted} if the symbols contained in $w$ occur in the left to right direction according to the total order of $\Omega$. For example, the word $012$ is sorted, but $021$ is not sorted. For any set of symbols $S$, we use the notation 
$$\wo(S)=\text{ the sorted word consisting of the symbols in } S.$$  For example, if $S=\{0,1,2\}$, then $\wo(S)=012$ and $\wo(\{2,0\})=02$.
\pssi
Let $g\in\Omega$ and $w$ be a string. The expression $|w|_g$ denotes the number of occurrences of $g$ in $w$, and the expression $\al w$ denotes the set $\{g\in\Omega: |w|_g>0\}$, that is, the set of symbols  that occur in $w$. For example, $$\al(1122010)=\{0,1,2\}.$$
\pssi
An \emdef{alphabet} is any finite nonempty subset of $\Omega$. In the following definitions we consider an alphabet $\Gamma$, called the alphabet of \emdef{reference}, and we assume that \emph{$\Gamma$ contains at least two symbols and no special symbols}. 

\pmsn\textbf{Algorithmic convention about alphabet symbols.} We shall consider algorithms on automata and transducers where the alphabets $\Gamma$ involved are not of fixed size and, therefore, $\card\Gamma\to\infty$; thus, the alphabet size $\card\Gamma$ is accounted for in time complexity estimates. Moreover, we assume that each $\Gamma$-symbol is of size $O(1)$.  This approach is also used in related literature (e.g., \cite{AllMoh:2003}), where it is assumed implicitly that the cost of comparing two $\Gamma$-symbols is $O(1)$. A similar assumption is made in graph algorithms where the  size of a graph $(V,E)$ is $\card V+\card E\to\infty$, but the size of each vertex is implicitly considered to be $O(1)$, \cite{Man:1989}. We note that there are proposals to represent the elements of $\Gamma$ using non-constant size objects---for instance, \cite{ADKN:2005} represents each $\Gamma$-symbol as a binary word of length $O(\log\card\Gamma)$.  
  
In the algorithms presented below, we need operations that require to access only a part of $\Gamma$ or some information about $\Gamma$ such as $|\Gamma|$. We assume that $\Gamma$ has been preprocessed such that the value of $|\Gamma|$ is available and is $O(\log|\Gamma|)$ bits long and  the \emdef{minimum symbol}   $\min\Gamma$ of $\Gamma$ is also available. In particular, we assume that we have available a \emph{sorted array} $\garr$ consisting of all $\Gamma$-symbols. While this is a convenient assumption, if in fact it is not applicable then one can make the array from $\Gamma$ in time $O\big(\card\Gamma\log\card\Gamma\big)$. Then, the minimum symbol of $\Gamma$ is simply $\garr[0]$.

Moreover, we have available an  \emdef{algorithm} $\outof(w)$, which returns a symbol in $\Gamma$ that is not in $\al w$, where $w$ is a \emph{sorted word} in $\Gamma^*$ with $0<|w|<|\Gamma|$. Next we explain that the desired algorithm 
\begin{equation}\mylabel{eq:costbound}
    \outof(w)\> \text{ can be made to work in time }\> O(|w|) 
\end{equation}
The algorithm $\outof(w)$ works by using an index $i$, initially $i=0$, and incrementing $i$ until $\garr[i]\neq w[i]$, in which case the algorithm returns $\garr[i]$.

\section{Set Specifications} \label{SEC:sym:sets}
Here we define expressions, called set specs, that are used to represent subsets of the alphabet $\Gamma$ or the empty string. These  can be used as labels in automata-type objects (labelled graphs) and regular expressions defined in subsequent sections. We also present some basic algorithms on set specs, which are needed for processing those regular expressions and labelled graphs.

\begin{definition}\mylabel{DEF:ssets}
A \emdef{set specification}, or \emdef{set spec} for short, is any string of one of the four forms
\[
 \ews \qquad \sany \qquad \sone w\qquad \snone w
\]
where $w$ is any sorted nonempty string containing no repeated symbols and no special symbols. The set of set specs is denoted by $\sspec$.  
\pnsi
Let $F,\sone u,\snone u,\sone v,\snone v$  be any set specs with $F\neq\ews$. We define the partial \emdef{operation} $\>\cap:\sspec\times\sspec\pto\sspec$  as follows.
\pssi $\ews\cap\ews=\ews$,\quad $\ews\cap F=F\cap\ews=\undef$
\pssi $\sany\cap F=F\cap\sany=F$
\pssi $\sone u\cap\sone v=\sone \wo\big(\al u\cap \al v\big)$,\quad if $\big(\al u\cap \al v\big)\neq\emptyset$
\pssi $\sone u\cap\sone v=\undef$,\quad if $\big(\al u\cap \al v\big)=\emptyset$
\pssi $\snone u\cap\snone v=\snone\wo\big(\al u\cup \al v\big)$
\pssi $\sone u\cap\snone v=\sone \wo\big(\al u\setminus \al v\big)$,\quad if $\big(\al u\setminus \al v\big)\neq\emptyset$
\pssi $\sone u\cap\snone v=\undef$,\quad if $\big(\al u\setminus \al v\big)=\emptyset$
\pssi $\snone u\cap\sone v=\sone v\cap\snone u$
\end{definition}

\begin{example}\mylabel{EX:ssets}
As any set spec $X$ is a string, it has a length $|X|$. We have that $\szabc{\sany}=1$ and $\szabc{\sone w}=1+\szabc w$. 
Also, 
\[
\sone035\cap\sone1358=\sone35, \quad
\snone035\cap\sone1358=\sone18, \quad
\snone035\cap\snone1358=\snone01358.
\]
\end{example}

\begin{lemma}\mylabel{LEM:ssets}
For any given set specs $G$ and $F$,  $G\cap F$ can be computed in time $O(\szabc G+\szabc F)$.	
\end{lemma}
\begin{proof}
The required algorithm works as follows. If either of $G,F$ is $\ews$ then $G\cap F$ is computed according to Def.~\ref{DEF:ssets}. Else, if either of $G,F$ is $\sany$,  return the other one. Now suppose that $G=\sone u$ and $F=\sone v$. As $u,v$ are sorted, the sorted word $w$ consisting of their common symbols is computed by using two indices $i$ and $j$, initially pointing to the first symbols of $u$ and $v$, and then advancing them as follows: if $u[i]=v[j]$ then output $u[i]$ and increment both $i,j$ by 1; if $u[i]<v[j]$  then increment only $i$; else increment only $j$. So the output would be $\sone w$, if $|w|>0$, or $\undef$ if $|w|=0$. In either case, each symbol of $u$ and $v$ is not accessed more than once, so the process works in time $O(\szabc u+\szabc v)$. Now suppose that $G=\snone u$ and $F=\snone v$. Then one can use a process similar to the above to compute the sorted word $w$ consisting of the union of the symbols in $u,v$. So the output would be $\snone w$. Now suppose that $G=\sone u$ and $F=\snone v$. Again the process to compute the sorted word $w$ consisting of the symbols in $u$ that are not in $v$ involves two indices $i$ and $j$, initially pointing to the first symbols of $u$ and $v$, and then advancing them as follows: if $u[i]=v[j]$ then increment both $i,j$ by 1; if $u[i]<v[j]$  then output $u[i]$ and increment only $i$; else increment only $j$. So the output would be $\sone w$, if $|w|>0$, or $\undef$ if $|w|=0$. The last case about $G,F$ is symmetric to the last one. 
\end{proof}

\begin{definition}\mylabel{DEF:sset:lang}
	Let $\Gamma$ be an alphabet of reference and let $F$ be a set spec. We say that $F$ \emdef{respects} $\Gamma$, if the following restrictions hold when $F$ is of the form $\sone w$ or $\snone w$:
	\pmsi $w\in\Gamma^* \>\text{ and }\> 0<|w|<|\Gamma|.$
	\pmsn
	In this case, the \emdef{language} $\lof F$ of $F$ (\emdef{with respect} to $\Gamma$) is the subset of $\Gamma\cup\{\ew\}$ defined as follows:
\[
\lof{\ews}=\{\ew\},\qquad 
\lof{\sany }=\Gamma,\qquad 
\lof{\sone w}=\al w,\qquad 
\lof{\snone w}=\Gamma\setminus\al w.
	\] 
The set of set specs that respect $\Gamma$ is denoted as follows
\[
\sspec[\Gamma]=\{\alpha\in\sspec\mid \alpha \text{ respects } \Gamma\}.
\]
\end{definition}

\begin{remark}\mylabel{REM:sset:lang}
	In the above definition, the requirement $|w|<|\Gamma|$ implies that there is at least one $\Gamma$-symbol that does not occur in $w$. Thus, to represent $\Gamma$ we must use $\sany$ as opposed to the longer set spec $\sone \wo(\Gamma)$. 	
\end{remark}

\begin{lemma}\mylabel{LEM:ssets:lang}
	Let $\Gamma$ be an alphabet of reference and let $G,F$ be set specs respecting $\Gamma$. The following statements hold true.
	\begin{enumerate}
    \item $\lof F\neq\eset$; and $\>\>\lof F=\Gamma$ if and only if $F=\sany$.
	\item $\lof{G\cap F}=\lof{G}\cap\lof{F}$, if $G\cap F\neq\undef$. 
	\item If $F=\sone w$ or $F=\snone w$ then $|\lof F|\le|\Gamma|-1$. 
	\item $\szabc{F}\le|\Gamma|$. 
	\end{enumerate} 
\end{lemma}
\begin{proof}
The first statement follows from the above definition and the following: If $F=\sone w$ then $(\al w)\notin\{\emptyset,\Gamma\}$, as $0<|w|<|\Gamma|$; and if  $F=\snone w$ then again $(\Gamma\setminus\al w)\notin\{\emptyset,\Gamma\}$. 
For the \underline{second} statement, we consider the definition of the operation`\,$\cap$\,' as well as the above definition. Clearly the statement holds, if  $G=F=\ews$, or if one of $G,F$ is $\sany$ and the other one is not $\ews$. Then, one considers the six cases of Definition~\ref{DEF:ssets} where $G,F$ contain $\sone$ or $\snone$. For example, if $G=\sone u$ and $F=\snone v$, we have that $\lof F = \Gamma\setminus\al v$, so  $\lof{G}\cap\lof{F}=\al u\setminus\al v$, which is equal to $\lof{G\cap F}$. The other cases can be shown analogously. 
The \underline{third} and \underline{fourth} statements follow from the restriction $0<|w|<|\Gamma|$ in Definition~\ref{DEF:sset:lang}. 
\end{proof}

The next lemma concerns simple algorithmic questions about set specs that are needed as basic tools in other algorithms further below.

\begin{lemma}\mylabel{LEM:ssets:algos}
	Let $\Gamma$ be an alphabet of reference and let $F\neq\ews$ be a set spec respecting $\Gamma$. The following statements hold true.
	\begin{enumerate}
    \item For given $g\in\Gamma$, testing whether $g\in\lof F$ can be done in time $O(\log\szabc F)$.
    \item For given $g\in\Gamma$, testing whether $\lof F\setminus \{g\}=\emptyset$ can be done in time $O(\szabc F)$.
    \item For any fixed $k\in\N$, testing whether $|\lof F|\ge k$ can be done in time $O(\szabc F+\log|\Gamma|)$, assuming the number $|\Gamma|$ is given as input along with $F$.
    \item Testing whether $|\lof F|=1$ and, in this case, computing the single element of $\lof F$ can be done in time $O(\szabc F)$.
    \item Computing an element of $\lof F$ can be done in time $O(\szabc F)$.
    \item If $\card{\lof F}\ge2$ then computing two different $\lof F$-elements can be done in time $O(\szabc{F})$.
	\end{enumerate} 
\end{lemma}
\begin{proof}
For the \underline{first} statement, we note that the condition to test is equivalent to one of ``$F=\sany$'', ``$F=\sone w$ and $|w|_g>0$'', ``$F=\snone w$ and $|w|_g=0$''; and that one can use binary search to test whether $g$ occurs in $w$.
For the \underline{second} statement, we note that the condition to test is equivalent to one of ``$F=\sone g$'', ``$F=\snone w$ and $|w|=|\Gamma|-1$ and $|w|_g=0$''. 
For the \underline{third} statement, we note that the condition to test is equivalent to one of $|\Gamma|\ge k$, $|w|\ge k$, $|\Gamma|-|w|\ge k$, depending on whether $F$ is one of $\sany,\sone w,\snone w$. The last case requires time $O(|F|)$ to compute $|w|$ and then $O(\log|\Gamma|+\log|w|)$ time for arithmetic operations, which is $O(\log|\Gamma|)$ as $|w|<|\Gamma|$. 
For the \underline{fourth} statement, we note that $|\lof F|=1$ is equivalent to whether ``$F=\sone g$ and $|g|=1$'' or ``$F=\snone w$ and $|w|=|\Gamma|-1$''. In the former case, the algorithm returns $g$. In the latter case, we use the algorithm $\outof(w)$ to get the desired symbol in $\Gamma\setminus\al w$. The latter case is the worse of the two,  and works in time $O(|F|+\log|\Gamma|)$ to compute $|w|$ and test whether $|w|=|\Gamma|-1$, plus time $O(|F|)$ to execute $\outof(w)$ (see the bound in \eqref{eq:costbound}). The total time is $O(|F|)$, as $|F|=|\Gamma|$. For the \underline{fifth} statement, if $F=\sany$ or $F=\sone w$ the algorithm simply returns $\garr[0]$ or $w[0]$, respectively. The worst case is when $F=\snone w$, where, as before, the algorithm uses $\outof(w)$ requiring time $(|F|)$. For the \underline{sixth} statement, the algorithm first finds any $g_1\in\lof F$, then computes the set spec $B=F\cap\snone g_1$ and then computes any $g_2\in\lof B$.
\end{proof}

\section{Pairing Specifications}\label{SEC:pairing}
Here we define  expressions for describing certain finite relations that are subsets of $(\Gamma\cup\{\ew\})\times(\Gamma\cup\{\ew\})$. First, we define their syntax and then their semantics. 

\begin{definition}\mylabel{DEF:spairs}
A \emdef{pairing specification}, or \emdef{pairing spec} for short, is a string of the form
\begin{equation}\label{eq:pairspec}
\ews/\ews\qquad 
\ews/G\qquad
F/\ews\qquad
F/G\qquad  
\ssame F \qquad 
\sdiff F G
\end{equation}
where $F,G$ are set specs with $F,G\neq\ews$. The set of pairing specs is denoted by $\srel$. The \emdef{inverse} $\sppinv$ of a pairing spec $\spp$ is defined as follows depending on the  possible forms of $\spp$ displayed in~\eqref{eq:pairspec}:
\pmsi	
$(\ews/\ews)^{-1}=(\ews/\ews),
\qquad (\ews/G)^{-1}=(G/\ews),\quad (F/\ews)^{-1}=(\ews/F),$
\pssi
$(F/G)^{-1}=(G/F),\quad(\ssame{F})^{-1}=(\ssame{F}),\quad
	(\sdiff{F}{G})^{-1}=(\sdiff{G}{F})$
\end{definition}

\begin{example}\mylabel{EX:spairs}
As a pairing spec $\spp$ is a string, it has a length $|\spp|$. We have that  $\szabc{\sany/\ews}=3$ and $\szabc{\sone u/\snone v}=3+\szabc u+\szabc v$. Also, $(\sany/\ews)^{-1}=(\ews/\sany)$ and $(\sdiff{\sone u}{\sany})^{-1}=(\sdiff{\sany}{\sone u})$.
\end{example}

\begin{definition}\mylabel{DEF:spairs:invar}
	A pairing spec is called  \emdef{alphabet invariant} if it contains no  set spec of the form $\sone w,\snone w$. The set of alphabet invariant pairing specs is denoted by $\srelinvar$. 
\end{definition}

\begin{definition}\mylabel{DEF:spairs:rel}
		Let $\Gamma$ be an alphabet of reference and let $\spp$ be a pairing spec. We say that $\spp$ \emdef{respects} $\Gamma$, if any set spec occurring in $\spp$ respects $\Gamma$.  The set of pairing specs that respect $\Gamma$ is denoted as follows
		\[\srel[\Gamma]=\{\spp\in\srel:\text{ $\spp$ respects $\Gamma$}\}.\]
		The \emdef{relation} $\rel{\spp}$ described by $\spp$ (with respect to $\Gamma$) is the subset of $\Gamma^*\times\Gamma^*$ defined as follows.
%
\pssi $\rel{\ews/\ews}=\{(\varepsilon,\varepsilon)\}$;
\pssi $\rel{\ews/G}=\{(\varepsilon,y)\mid y\in \lof G\}$;
\pssi $\rel{F/\ews}=\{(x,\varepsilon)\mid x\in \lof F\}$;
\pssi $\rel{F/G}=\{(x,y)\mid x\in \lof F,y\in \lof G\}$;
\pssi $\rel{\ssame F}=\{(x,x)\mid x\in \lof F\}$;
\pssi $\rel{\sdiff F G}=\{(x,y)\mid x\in \lof F,y\in\lof G,x\not=y\}$.
\end{definition}

\begin{remark}\mylabel{REM:alphainvar}
	All the alphabet invariant pairing specs are
	\[
	\ews/\ews \quad\quad \ews/\sany \quad\quad \sany/\ews \quad\quad \sany/\sany \quad\quad \sany/{\bm=} \quad\quad \sany/\sany{\bm\neq} 
	\]
	Any alphabet invariant pairing spec $\spp$ respects all alphabets of reference, as $\spp$ contains no set specs of the form $\sone w$ or $\snone w$.
\end{remark}

\begin{lemma}\mylabel{LEM:sym:inverse}
	Let $\spp\in\srel[\Gamma]$. The following statements hold true.
	\begin{enumerate}
	\item $\rel{\spp}=\emptyset$ if and only if 
	 $\spp$ is of the form $\sdiff{F}{G}$ and $\lof F=\lof G=\{g\}$ for some $g\in\Gamma$.
    \item $\rel{\sppinv}=\rel{\spp}^{-1}$.
    \item $\sppinv$ can be computed from $\spp$ in time $O(\szabc\spp)$.
    \end{enumerate}	
\end{lemma}
\begin{proof}
	The \underline{first} statement follows from Definition~\ref{DEF:spairs:rel}  when we note that the set $\{(x,y)\mid x\in \lof F,y\in\lof G,x\not=y\}$ is empty if and only if $\lof F=\lof G=\{g\}$, for some $g\in\Gamma$. The \underline{last} two statements follow from Definitions \ref{DEF:spairs} and \ref{DEF:spairs:rel}. 
\end{proof}

\pnsn\textbf{Some notation on pairing specs.} Let $\spp$ be a pairing spec. Then the \emdef{left} part, $\spleft\spp$, of $\spp$ is the string on the left of the symbol `/', and the \emdef{right} part, $\spright\spp$, of $\spp$ is the string on the right of `/'. We have the following examples:
\[
\spleft(\sdiff{\sone w}{\sany}) = \sone w\quad \spright(\sdiff{\sone w}{\sany}) = \sany{\bm\neq}\quad
\spleft(\ssame\sany) = \sany\quad \spright(\ssame\sany) = {\bm=}
\]
While the expression $\lof{\spleft\spp}$  makes sense when $\spp$ respects the alphabet of reference, this is not the case for $\lof{\spright\spp}$. So we define $\sprset\spp$ to be as follows, depending on the structure of $\spp$ according to \eqref{eq:pairspec}
\pssi
$\sprset(\ews/\ews)=\ews,\qquad\sprset(\ews/G)=G,\quad\sprset(F/\ews)=\ews,$
\pssi
$\sprset(F/G)=G,\quad\sprset(\ssame F)=F,\quad\sprset(\sdiff F G)=G.$
\pssn The above notation implies 
\begin{equation}\label{EQ:leftright}
	\rel{\spp}\subseteq\lof{\spleft\spp}\times\lof{\sprset\spp}.
\end{equation}


\begin{lemma}\mylabel{LEM:sym:restrict}
	If $\spp\in\srel[\Gamma]$ then $\szabc{\spp}\le2|\Gamma|+2$.
\end{lemma}
\begin{proof}
	Follows from Lemma~\ref{LEM:ssets:lang}. 
\end{proof}

\section{Label Sets and their Behaviours}\label{SEC:labelsets}\prlabel{SEC:labelsets}
We are interested in automata-type objects (labelled graphs) $\autg$ in which every transition label $\beta$ represents a set $\cali(\beta)$ of elements in some monoid $M$. The subsets $\cali(\beta)\sseq M$ are the behaviours  of the labels and they are used to define the behaviour of $\autg$ as a subset of $M$. We focus on sets of labels in this section---see next section for labelled graphs. We shall use the notation
\pssi $\ew_M$ for the neutral element of the monoid $M$.
\pssn
If $S,S'$ are any two subsets of $M$ then, as usual, we define 
\[
SS'=\{mm'\mid m\in S,\> m'\in S'\}\quad\text{ and }\quad S^i=S^{i-1}S\quad \text{ and }\quad S^*=\cup_{i=0}^\infty S^i.
\]
where $S^0=\{\ew_M\}$ and the monoid operation is denoted by simply concatenating elements. We shall only consider \emdef{finitely generated} monoids $M$ where each $m\in M$ has a \emdef{canonical} (string) representation $\underline m$. Then, we write
\[\underline M=\{\underline m\mid m\in M\}.\]
In the example below, we provide sample canonical representations for the two monoids of interest to this work.

\begin{example}\label{EX:stand:monoids}\prlabel{EX:stand:monoids}
    We shall consider two standard monoids. 
    \begin{enumerate}
    \item The free monoid $\Gamma^*$ (or $\Sigma^*$) whose neutral element is $\ew$. The canonical representation of a nonempty word $w$ is $w$ itself and that of $\ew$ is $\ews$: $\underline\ew=\ews$. 
    \item The monoid $\Sigma^*\times\Delta^*$ (or $\Gamma^*\times\Gamma^*$) whose neutral element is $(\ew,\ew)$. The canonical representation of a word pair $(u,v)$ is $\underline u/\underline v$. In particular, $\underline{(\ew,\ew)}=\ews/\ews$.
    \end{enumerate}
\end{example}

A \emdef{label set} $B$ is a nonempty set of nonempty strings (over $\Omega$). A \emdef{label behaviour} is a mapping 
$$\cali:B\to 2^M,$$ where $M$ is a monoid. Thus, the behaviour $\beh{\beta}$ of a label $\beta\in B$ is a subset of $M$. We shall  consider label sets $B$ with \emdef{fixed behaviours}, so we shall 
\pssi denote by $\>\mon B\>$  the \emdef{monoid of} $B$ via its fixed behaviour.

\pssn
\textbf{Notational Convention.}
We shall  make the \emdef{convention} that for any label sets $B_1,B_2$ with fixed behaviours $\cali_1,\cali_2$, we have:
\pssi if  
$\mon B_1=\mon B_2$ then  
$\cali_1(\beta)=\cali_2(\beta)$, for all $\beta\in B_1\cap B_2$.
\pssn  
With this convention we can simply use a single behaviour notation $\cali$ for all label sets with the same behaviour monoid, that is, we shall use $\cali$ for any $B_1, B_2$ with $\mon B_1=\mon B_2$. This convention is applied in the example below: we use $\call$ for the behaviour of both the label sets  $\Sigma_{\ews}$ and $\sspec[\Gamma]$.

\begin{example}\label{EX:fixedbeh}\prlabel{EX:fixedbeh}
We shall use some of the following label sets and their  fixed label behaviours.
\begin{enumerate}
	\item $\Sigma_{\ews}=\Sigma\cup\{\ews\}$ with behaviour $\call:\Sigma_{\ews}\to 2^{\Sigma^*}$ such that $\lof g=\{g\}$, if $g\in\Sigma$, and $\lof\ews=\{\ew\}$. Thus, $\mon\Sigma=\Sigma^*$.
	\item $\Sigma$ with behaviour $\call:\Sigma\to 2^{\Sigma^*}$ such that $\lof g=\{g\}$, for $g\in\Sigma$. Thus, $\mon\Sigma=\Sigma^*$.
	\item $\sspec[\Gamma]$ with behaviour $\call: \sspec[\Gamma]\to 2^{\Gamma^*}$, as specified in Def.~\ref{DEF:sset:lang}. Thus, $\mon\sspec[\Gamma]=\Gamma^*$.
	\item $\REG\Sigma=\REG\Sigma_{\ews}$ = all regular expressions over $\Sigma$ with behaviour $\call: \REG\Sigma\to 2^{\Sigma^*}$ such that $\lof r$ is the language of the regular expression $r$. Thus, $\mon(\REG\Sigma)=\Sigma^*$.
	\item $[\Sigma_{\ews},\Delta_{\ews}]=\{x/y\mid x\in\Sigma_{\ews},y\in\Delta_{\ews}\}$ with behaviour 
	$$\calr: [\Sigma_{\ews},\Delta_{\ews}]\to 2^{\Sigma^*\times\Delta^*}$$ such that $\rel{\ews/\ews}=\{(\ew,\ew)\}$, $\rel{x/\ews}=\{(x,\ew)\}$, $\rel{\ews/y}=\{(\ew,y)\}$, $\rel{x/y}=\{(x,y)\}$,  for any $x\in\Sigma$ and $y\in\Delta$. Thus, $\mon[\Sigma_{\ews},\Delta_{\ews}]=\Sigma^*\times\Delta^*$. 
	\item $\srel[\Gamma]$ with behaviour $\calr: \srel[\Gamma]\to 2^{\Gamma^*\times\Gamma^*}$ as specified in Def.~\ref{DEF:spairs:rel}. Thus, $\mon\srel[\Gamma]=\Gamma^*\times\Gamma^*$.
	\item $\srelinvar$ with behaviour $\calr_{\undef}:\srelinvar\to\{\eset\}$. Thus, $\cali(\beta)=\eset$, for any $\beta\in\srelinvar$.
	\item If $B_1,B_2$ are label sets with behaviours $\cali_1,\cali_2$, respectively, then $[B_1,B_2]$ is the label set $\{\beta_1/\beta_2\mid \beta_1\in B_1,\beta_2\in B_2\}$ with behaviour and monoid such that  
	    $$\cali(\beta_1/\beta_2) = \cali_1(\beta_1)\times\cali_2(\beta_2)\quad\text{and}\quad \mon[B_1,B_2]=\mon B_1\times\mon B_2.$$
	\item $[\reg\Sigma,\reg\Delta]$ with behaviour $\calr$ in the monoid $\Sigma^*\times\Delta^*$ such that $\rel{\rer/\res}=\lof{\rer}\times\lof{\res}$, for any $\rer\in\reg\Sigma$ and $\res\in\reg\Delta$.
\end{enumerate}
For any monoid of interest $M$, $\underline M$ is a label set such that $$\mon \underline M=M \quad\text{and}\quad \cali(\underline m)=\{m\}.$$
Thus for example, as $\mon\srel[\Gamma]=\mon\underline{\Gamma^*\times\Gamma^*}=\Gamma^*\times\Gamma^*$ and the behaviour of $\srel$ is denoted by $\calr$, we have  $\rel{\underline{(0,1)}}=\rel{0/1}=\{(0,1)\}=\rel{\sone0/\sone1}$.
\end{example}

\begin{remark}\label{REM:labelsets}\prlabel{REM:labelsets}
	We shall not attempt to define the set of all labels. We limit ourselves to those of interest in this paper. Of course one can define new label sets $X$ at will, depending on the application; and in doing so, one would also define concepts related to those label sets, such as the $\mon X$.
\end{remark}

\section{Labelled Graphs, Automata, Transducers}\label{SEC:graphs}\prlabel{SEC:graphs}
Let $B$ be a label set with behaviour $\cali$. A \emdef{type $B$   graph} is a quintuple 
$$\autg = \big(Q,B,\delta,I,F\big)$$ 
such that 
\begin{itemize}
  \item $Q$ is a  nonempty set whose elements are called \emdef{states};
  \item  $I\sseq Q$ is the nonempty set of initial, or start states;
  \item  $F\sseq Q$ is the set of final states;
  \item  $\delta$ is a set, called the set of \emdef{edges} or \emdef{transitions}, consisting of triples $(p,\beta,q)$ such that $p,q\in Q$ and $\beta$ is a nonempty string of $\Omega$-symbols. 
  \item the set of \emdef{labels}  
	$\Labels(\autg) =\{\beta\mid (p,\beta,q)\in\delta\}$
	is a subset of $B$.
\end{itemize}

We shall use the term \emdef{labelled graph} to mean a type $B$ graph as defined above, for some label set $B$. The labelled graph is called \emdef{finite} if $Q$ and $\delta$ are both finite. \emph{Unless otherwise specified, a labelled graph, or type $B$ graph, will be assumed to be finite.}
\pnsi
As a label $\beta$ is a string, the length $\szabc\beta$ is well-defined. Then, the  \emdef{size} $\szabc{e}$ of an edge $e=(p,\beta,q)$ is the quantity $1+\szabc\beta$ and the size of $\delta$ is $\szw{\delta} =\sum_{e\in \delta}\sz{e}$. Then the \emdef{graph size} of $\autg$ is the quantity
\[
\szabc{\autg}=|Q|+\szw{\delta}.
\]
	\pnsi 
	A \emdef{path} $P$ of $\autg$ is a sequence of consecutive transitions, that is, $P=\sequ{q_{i-1},\beta_i,q_i}_{i=1}^\ell$ such that  each $(q_{i-1},\beta_i,q_i)$ is in $\delta$. The path $P$ is called \emdef{accepting}, if $q_0\in I$ and $q_\ell\in F$. If $\ell=0$ then $P$ is empty and it is an accepting path if $I\cap F\neq\eset$. 
	\pnsi
	A state is called \emdef{isolated}, if it does not occur in any transition of $\autg$. A state is called \emdef{useful}, if it occurs in some accepting path. Note that any state  in $I\cap F$ is useful and can be isolated. The labelled graph $\autg$ is called \emdef{trim}, if 
	\pssi every state of $\autg$ is useful, and
	\pnsi $\autg$ has at most one isolated state in $I\cap F$. 
	\pssn Computing the trim part of $\autg$ means removing the non-useful states and keeping only one isolated state in $I\cap F$ (if such states exist), and can be computed in linear time $O(\szabc{\autg})$.

\begin{lemma}\label{LEM:trim}\prlabel{LEM:trim}
	Let $\autg=(Q,B,\delta,I,F)$ be a trim  labelled graph. We have that
	$$\card Q\le2\card\delta+1.$$
\end{lemma}
\begin{proof}
	$Q$ can be partitioned into three sets: $Q_1$ = the set of states having an outgoing edge but no incoming edge; $Q_2$ = the set of states having an incoming edge; and possibly a single isolated state in $I\cap F$. The claim follows from the fact that $\card{Q_1},\card{Q_2}\le\card{\delta}$.
\end{proof}

\begin{definition}\label{DEF:graph:behave}\prlabel{DEF:graph:behave}
	Let $\autg=\big(Q,B,\delta,I,F\big)$ be a  labelled  graph, for some label set $B$ with behaviour $\cali$. We define the \emdef{behaviour} $\cali(\autg)$ of $\autg$ as the set of all $m\in \mon B$ such that there is an accepting path $\sequ{q_{i-1},\beta_i,q_i}_{i=1}^\ell$ of $\autg$ with 
	\[m\in\cali(\beta_1)\cdots\cali(\beta_\ell). \]
	The \emdef{expansion} $\exp\autg$ of $\autg$ is the labelled  graph $\big(Q,\underline{\mon B},\delta_{\exp},I,F\big)$ such that
	\[
    \delta_{\exp}=\{(p,\underline m,q)\mid \exists\,(p,\beta,q)\in\delta: m\in\cali(\beta) \}.
     \] 
     In some cases it is useful to modify $\autg$ by adding the transition $(q,\srep{\nel B},q)$ (a self loop) for each state $q$ of $\autg$. The resulting labelled graph is \emdef{denoted} by $\autg^{\ew}$.
\end{definition}

\begin{remark}\label{REM:canonical}\prlabel{REM:canonical}
	The above definition remains valid with no change if the labelled graph, or its expansion, is not finite. The expansion graph of $\autg$ can have infinitely many transitions---for example if $\autg$ is of type $\REG\Sigma$.
\end{remark}

\begin{lemma}\label{LEM:expanded}\prlabel{LEM:expanded}
	For each type $B$ graph $\autg=(Q,B,\delta,I,F)$, we have that 
	$$\cali(\autg)=\cali(\exp\autg)\quad \text{and} \quad \cali(\autg)=\cali(\autg^\ew).$$
\end{lemma}
\begin{proof}
	Let $m\in\cali(\exp\autg)$. Then there is an accepting path $\sequ{q_{i-1},\underline{m_i},q_i}_{i=1}^\ell$ of $\exp\autg$ such that $m\in\cali(\underline{m_1})\cdots\cali(\underline{m_\ell})=\{m_1\}\cdots\{m_\ell\}$; hence, $m=m_1\cdots m_\ell$. By definition of $\delta_{\exp}$, for each $i=1,\ldots,\ell$, there is $(q_{i-1},\beta_i,q_i) \in\delta$ such that $m_i\in\cali(\beta_i)$, so $\sequ{q_{i-1},\beta_i,q_i}_{i=1}^\ell$ is an accepting path of $\autg$. Then, $\cali(\beta_1)\cdots\cali(\beta_\ell)\sseq\cali(\autg)$, and $m\in\cali(\autg)$. Conversely, for any $m\in\cali(\autg)$, one uses similar arguments to show that $m\in\cali(\exp\autg)$. Thus, $\cali(\autg)=\cali(\exp\autg)$.
	\pnsi 
	To show that $\cali(\autg)=\cali(\autg^\ew)$, let $\delta^\ew$ be the set of transitions in $\autg^\ew$. As $\delta\sseq\delta^\ew$, we have $\cali(\autg)\sseq\cali(\autg^\ew)$. For the converse, the main idea is that, if any accepting path of $\autg^\ew$ contains  transitions $(q_{i-1},\srep{\nel B}, q_i)$ with $q_{i-1}=q_i$, then these transitions can be omitted resulting into an accepting path of $\autg$.  
\end{proof}

\begin{remark}\label{REM:eloops}\prlabel{REM:eloops}
	As stated before, our focus is on two kinds of monoids: $\Sigma^*$ and $\Sigma^*\times\Delta^*$. Recall that, in those monoids, the neutral elements $\ew$ and $(\ew,\ew)$ have canonical representations $\ews$ and $\ews/\ews$, which are of fixed length. Thus, we shall assume that $\underline{\ew_{\mon B}}=O(1)$, for any label set $B$. This implies that
	\[\szabc{\autg^\ew}=O(\szabc{\autg}).\]
\end{remark}

\begin{definition}\label{DEF:fsms}\prlabel{DEF:fsms} Let $\Sigma,\Delta,\Gamma$ be alphabets.
    \begin{enumerate}
    \item A nondeterministic finite \emdef{automaton} with empty transitions, or \emdef{$\ew$-NFA} for short, is a labelled graph $\auta=(Q,\Sigma_{\ews},\delta,I,F)$. If $\Labels(\auta)\sseq\Sigma$ then $\auta$ is called an \emdef{NFA}. The \emdef{language $\lof\auta$} accepted by $\auta$ is the behaviour of $\auta$ with respect to the label set $\Sigma_{\ews}$.
    \item An \emdef{automaton with set specs} is a  labelled graph $\autb=(Q,\sspec[\Gamma],\delta,I,F)$. The \emdef{language $\lof\autb$} accepted by $\autb$ is the behaviour of $\autb$ with respect to the label set $\sspec[\Gamma]$.
    \item A \emdef{transducer (in standard form)} is a  labelled graph $\trt=(Q,[\Sigma_{\ews},\Delta_{\ews}],\delta,I,F)$. The \emdef{relation $\rel\trt$ realized} by $\trt$ is the behaviour of $\trt$ with respect to the label set $[\Sigma_{\ews},\Delta_{\ews}]$.
    \item A \emdef{transducer with set specs} is a  labelled graph  $\trs=(Q,\srel[\Gamma],\delta,I,F)$. The \emdef{relation $\rel\trs$ realized} by $\trs$ is the behaviour of $\trs$ with respect to the label set $\srel[\Gamma]$.
    \item An \emdef{alphabet invariant transducer} is a  labelled graph $\tri  = ( Q, \srelinvar,\delta, I, F )$. If $\Gamma$ is an alphabet then the \emdef{$\Gamma$-version} of $\tri$ is the transducer with set specs $\tri[\Gamma]  = ( Q, \srel[\Gamma], \delta, I, F ).$ 
    \end{enumerate}
\end{definition}

\begin{remark}
    The above definitions about automata and transducers are equivalent to the standard ones. The only slight  	deviation is that, instead of using the empty word $\ew$ in transition labels, here we use the empty word symbol $\ews$. This has two advantages: (i) it allows us to make a uniform presentation of definitions and results and (ii) it is consistent with the use of a symbol for the empty word in regular expressions. As usual about transducers $\trt$, we denote by $\trt(w)$ the \emdef{set of outputs of $\trt$} on input $w$, that is, 
    $$\trt(w)=\{u\mid (w,u)\in\rel\trt\}.$$
     Moreover, for any language $L$, we have that $\trt(L)=\cup_{w\in L}\trt(w)$.
\end{remark}

\begin{remark}
	The size of an alphabet invariant transducer $\tri$ is of the same order of magnitude as $|Q|+|\delta|$.
\end{remark}

\begin{lemma}\label{LEM:exp}\prlabel{LEM:exp}
	If $\autb$ is an automaton with set specs then $\exp\autb$ is an $\ew$-NFA. If $\trs$ is a transducer with set specs then $\exp\trs$ is a transducer (in standard form).
\end{lemma}

\pnsn\textbf{Convention.} Let $\Phi(\tru)$ be any statement about the behaviour of an automaton or transducer $\tru$. If $\trv$ is an automaton or transducer with set specs then we make the convention that the statement $\Phi(\trv)$ means $\Phi(\exp\trv)$. For example, ``$\trs$ is an input-altering transducer'' means that ``$\exp\trs$ is an input-altering transducer''---a transducer $\trt$ is \emdef{input-altering} if $u\in\trt(w)$ implies $u\neq w$, or equivalently $(w,w)\notin\rel{\trt}$, for any word $w$.

\begin{example}\label{EX:sym:transd}\prlabel{EX:sym:transd}
The transducers shown in Fig.~\ref{FIG:sym:transd} are alphabet invariant. Both transducers are much more succinct compared to their expanded $\Gamma$-versions, as $|\Gamma|\to\infty$:
\[
|\exp\trt_{\rm sub2}[\Gamma]|=O(|\Gamma|^2)\quad\text{ and }\quad |\exp\trt_{\rm px}[\Gamma]|=O(|\Gamma|).
\]
	\begin{figure}\begin{center}
		\begin{tikzpicture}[>=stealth', shorten >=1pt, auto, node distance=2.50cm,initial text={}]
			\tikzset{sstate/.style={state,inner sep=3pt, minimum size=6pt}}
			\node [initial, sstate] (i) {$0$};
			\node [left of=i, node distance=1.25cm] () {$\trt_{\rm sub2}:$};
			\node [sstate, accepting] (s) [right of=i] {$1$};
			\node [sstate, accepting] (f) [right of=s] {$2$};
			\path[->] (i) edge node {$\sdiff{\sany}{\sany}$} (s)
			          (s) edge node {$\sdiff{\sany}{\sany}$} (f)
			          (i) edge [loop above] node [above] {$\ssame{\sany}$} ()
			          (s) edge [loop above] node [above] {$\ssame{\sany}$} ()
			          (f) edge [loop above] node [above] {$\ssame{\sany} $} ()
			;
			\node [initial, sstate, right of=f] (i1) {$0$};
			\node [left of=i1, node distance=1.00cm] () {$\trt_{\rm px}:$};
			\node [sstate, accepting] (f1) [right of=i1] {$1$};
			\path[->] (i1) edge node {${\sany}/{\ews}$} (f1)
			          (i1) edge [loop above] node [above] {$\ssame{\sany}$} ()
			          (f1) edge [loop above] node [above] {$\ssame{\sany} $} ()
			;
		\end{tikzpicture}
		\pmsn\parbox{0.85\textwidth}{\caption{The left transducer realizes the relation  of all $(u,v)$ such that $u\neq v$ and the Hamming distance of $u,v$ is at most 2. The right transducer realizes the relation of all $(u,v)$ such that $v$ is a proper prefix of $u$.}\label{FIG:sym:transd}\prlabel{FIG:sym:transd}}
	\end{center}\end{figure}
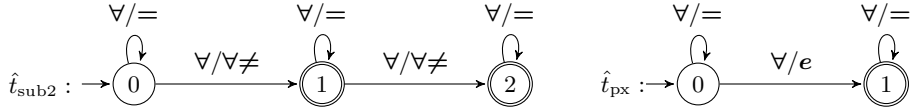
\end{example}

\begin{example}\label{EX:sym:autom}\prlabel{EX:sym:autom}
If expanded, the automaton with set specs in Fig.~\ref{FIG:sym:autom}, will have $3n+1$ transitions, as opposed to the current 7 ones. 
	\begin{figure}\begin{center}
		\begin{tikzpicture}[>=stealth', shorten >=1pt, auto, node distance=2.50cm,initial text={}]
			\tikzset{sstate/.style={state,inner sep=3pt, minimum size=6pt}}
			\node [initial, accepting, sstate] (i) {};
			\node [left of=i, node distance=1.25cm] () {$\autb:$};
			\node [sstate, accepting] (s) [right of=i] {};
			\node [sstate, accepting] (f) [right of=s] {};
			\path[->] (i) edge node [below] {$0$} (s)
			          (s) edge node [below] {$1$} (f)
			          (s) edge [bend right=35] node [above] {$\snone01$} (i)
			          (f) edge [bend left=35] node [above] {$\snone01$} (i)
			          (f) edge [bend right=35] node [above] {$0$} (s)
			          (i) edge [loop above] node [above] {$\snone0$} ()
			          (s) edge [loop above] node [above] {$0$} ()
			;
		\end{tikzpicture}
		\pmsn\parbox{0.85\textwidth}{\caption{The automaton accepts all strings over $\Gamma=\{0,\ldots,n\}$ that do not contain 011.}\label{FIG:sym:autom}\prlabel{FIG:sym:autom}}
	\end{center}\end{figure}
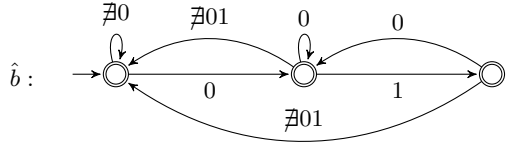
\end{example}

Following \cite{Yu:handbook}, if $\trt=(Q,[\Sigma_{\ews},\Delta_{\ews}],\delta,I,F)$ is a transducer then $\trtinv$ is the transducer $(Q,[\Delta_{\ews},\Sigma_{\ews}],\delta',I,F)$, where 
$\delta'=\{(p,y/x,q)\mid (p,x/y,q)\in\delta\},$ such that 
\begin{equation}\label{EQ:trinv}\prlabel{EQ:trinv}
\rel{\trtinv}=\rel{\trt}^{-1}.	
\end{equation}

\begin{lemma}\mylabel{LEM:expanded}
	For each transducer $\trs$ with set specs we have that 
	$$
	\exp(\trsinv)=(\exp\trs)^{-1}\quad\text{and}\quad\rel{\trsinv}=\rel{\trs}^{-1}.$$
\end{lemma}
\begin{proof}
    The first identity follows from two facts: (i) $\exp(\trsinv)$ has transitions of the form $(p,y/x,q)$, where $y/x\in\rel{\sppinv}$ and $(p,\sppinv,q)$ is a transition in $\trsinv$; and (ii) $(\exp\trs)^{-1}$ has transitions of the form $(p,y/x,q)$, where $x/y\in\rel{\spp}$ and $(p,\spp,q)$ is a transition in $\trs$.
	The second identity follows from \eqref{EQ:trinv} and Definition~\ref{DEF:graph:behave}.
\end{proof}

\begin{remark}\mylabel{REM:deltaszw}
	Let $\trt=(Q,\Gamma,\delta,I,F)$ be a transducer with set specs. By Lemma~\ref{LEM:sym:restrict}, we have that 
	$$\szw{\delta} \le (2\card\Gamma+3)\card\delta.$$
\end{remark}

\section{Rational Operations}\mylabel{SEC:rational}
The three standard \emdef{rational operations} (union, catenation, star) on automata and transducers can be defined on  labelled graphs with appropriate constraints on the monoids involved. Let $\autg=(Q,B,\delta,I,F)$ and $\autg'=(Q',B',\delta',I',F')$ be  labelled  graphs such that 
$$\mon B=\mon B'\quad\text{ and }\quad Q\cap Q'=\eset.$$

The graph $\autg\cup\autg'$ of type $C=B\cup B'\cup\{\srep{\nel B}\}$ is defined as follows
$$\autg\cup\autg'=(Q\cup Q'\cup\{s\},C,\delta\cup\delta'\cup E,\{s\},F\cup F'),$$
where $s$ is a (new) state not in $Q\cup Q'$ and $E$ is the set of transitions $(s,\srep{\nel B},p)$, for all $p\in I\cup I'$.  

The graph $\autg\cdot\autg'$ of type $C=B\cup B'\cup\{\srep{\nel B}\}$ is defined as follows
$$\autg\cdot\autg'=(Q\cup\{q\}\cup Q',C,\delta\cup\delta'\cup E\cup E',I,F'),$$
where $q$ is a (new) state not in $Q\cup Q'$, $E$ is the set of transitions $(f,\srep{\nel B},q)$, for all $f\in F$, and $E'$ is the set of transitions $(q,\srep{\nel B},i')$, for all $i'\in I'$.  

The graph $\autg^*$ of type $D=B\cup\{\srep{\nel B}\}$ is defined as follows
$$\autg^*=(Q\cup \{s\},D,\delta\cup E_1\cup E_2,\{s\},F\cup \{s\}),$$
where $s$ is a (new) state not in $Q\cup Q'$, $E_1$ is the set of transitions $(s,\srep{\nel B},i)$ for all $i\in I$, and $E_2$ is the set of transitions $(f,\srep{\nel B},s)$ for all $f\in F$.


\begin{lemma}\label{LEM:rat:op}\prlabel{LEM:rat:op}
	Let $\autg=(Q,B,\delta,I,F)$ and $\autg'=(Q',B',\delta',I',F')$ be  trim  labelled graphs such that $\mon B=\mon B'$. 
	\begin{enumerate}
		\item $\cali(\autg\cup\autg')=\cali(\autg)\cup\cali(\autg')$ and $\sz{\autg\cup\autg'}=O(\sz\autg+\sz{\autg'})$.
		\item $\cali(\autg\cdot\autg')=\cali(\autg)\cali(\autg')$ and $\sz{\autg\cdot\autg'}=O(\sz\autg+\sz{\autg'})$.
		\item $\cali(\autg^*)=\cali(\autg)^*$ and $\sz{\autg^*}=O(\sz\autg)$.
	\end{enumerate}
\end{lemma}

In the above lemma, the statements about the sizes of the graphs follow immediately from the definitions of their constructions. For the statements about the behaviours of the constructed graphs, it is sufficient to show the statements about their expansions. For example, for the third statement, one shows that
$$\cali(\exp\autg^*)=\cali(\exp\autg)^*.$$
But then, one works at the level of the monoid $\mon B$ and the proofs are essentially the same as the ones for the case of $\ew$-NFAs (see e.g. \cite{Wood:theory:of:comput}).

\section{Regular Expressions over Label Sets} \mylabel{SEC:sym:RE}
We extend the definitions of regular and 2D regular expressions to include set specs and pairing specs, respectively. We start off with a definition that would work with any label set (called set of atomic formulas in \cite{Sak:2015}).

\begin{definition}\label{DEF:reg:gen}\prlabel{DEF:reg:gen}
Let $B$ be a label set with behaviour $\cali$ such that no $\beta\in B$ contains the special symbol $\erel$. The set $\reg B$ of \emdef{type $B$ regular expressions} is the set of strings  consisting of the 1-symbol string $\erel$ and the strings in the set $Z$ that is defined inductively as follows.
\begin{itemize}
	\item $\underline{\ew_{\mon B}}$  is in $Z$. 
	\item Every $\beta\in B$ is in $Z$. 
	\item If $\rer,\res\in Z$ then $(\rer+\res), (\rer\cdot\res), (\rer^*)$ are in $Z$. 
\end{itemize}
The behaviour $\cali(\rer)$ of a type $B$ regular expression $\rer$ is defined inductively as follows.
\begin{itemize}
	\item $\cali(\erel)=\eset$ and $\cali(\underline{\ew_{\mon B}}) = \ew_{\mon B}$;
	\item $\cali(\beta)$ is the subset of $\mon B$ already defined by the behaviour $\cali$ on $B$;
	\item $\cali(\rer+\res)=\cali(\rer)\cup\cali(\res)$;
	\item $\cali(\rer\cdot\res)=\cali(\rer)\cali(\res)$;
	\item $\cali(\rer^*)=\cali(\rer)^*$. 
\end{itemize}
\end{definition}

\begin{example}\label{EX:reg}\prlabel{EX:reg}
	Let $\Sigma,\Delta$ be alphabets. Using $\Sigma$ as a label set, we have that $\reg \Sigma$ is the set of ordinary regular expressions over $\Sigma$. For the label set $[\Sigma_{\ews},\Delta_{\ews}]$, we have that $\reg[\Sigma_{\ews},\Delta_{\ews}]$ is the set of rational expressions over $\Sigma^*\times\Delta^*$ in the sense of \cite{Sak:2015}. 
\end{example}

\begin{example}\label{EX:res}\prlabel{EX:res}
    Let $\Gamma=\{0,1,\ldots,n-1\}$. In type $\sspec[\Gamma]$ regular expressions, the set specs $\sany ,\sone w,\snone w$ correspond to the following UNIX expressions, respectively: `.', `[$w$]', `[\verb1^1$w$]'. So $\lof \sany=\Gamma$.  When the alphabet size $n$ is a parameter rather than fixed, the savings when using expressions over label sets could be of order $O(n)$ or even $O(n^2)$. For example, the expression $\sany $ is of size $O(1)$ but the corresponding (ordinary) regular expression of type $\Gamma_{\ews}$ is $0+\cdots+(n-1)$, which is of size $O(n)$. Similarly, the following regular expression over $\srel[\Gamma]$ 
	\begin{equation}\mylabel{}
	(\ssame{\sany})^*\;\big(\sdiff{\sany}{\sany}\big)\;(\ssame{\sany})^*
	\end{equation}
	is of size $O(1)$. It describes all word pairs $(u,v)$ such that the Hamming distance of $u,v$ is 1. The corresponding (ordinary) regular expression over $\editops[\Gamma]$ is
	\[
	\big(0/0+\cdots+ (n-1)/(n-1)\big)^*\;\big(\rer_0+\cdots+ \rer_{n-1}\big)\;\big(0/0+\cdots +(n-1)/(n-1)\big)^*
	\]
	which is of size $O(n^2)$, where each $\rer_i$ is the sum of all expressions $i/j$ with $j\not=i$ and $i,j\in\Gamma$.
\end{example}

\begin{example}
    Consider the UNIX utility $\mathtt{tr}$. For any  strings $u,v$ of length $\ell>0$, the command
    \pssi $\mathtt{tr}\>\>u\>\>v$
    \pssn can be `simulated' by the following regular expression of type $\srel[\mathrm{ASCII}]$
    \pssi $\Big((\ssame{\snone u})\,+\>(\sone u[0]/\sone v[0])\,+\cdots+\,\big(\sone u[\ell-1]/v[\ell-1]\big)\Big)^*$
    \pssn where $\mathrm{ASCII}$ is the alphabet of standard ASCII characters. Similarly, the command
    \pssi $\mathtt{tr}\>\mathtt{-d}\>\>u$
    \pssn can be `simulated' by the following regular expression of type $\srel[\mathrm{ASCII}]$
    \pssi $\big(\sone u/\ews+\ssame{\snone u}\big)^*$
    \pssn For the command
    \pssi $\mathtt{tr}\>\mathtt{-s}\>\>u$
    \pssn it seems that any regular expression over $\srel[\mathrm{ASCII}]$ cannot be of size $O(\ell^2)$. A related (ordinary) transducer is shown below in Fig.~\ref{FIG:unix}.
	\begin{figure}\begin{center}
		\begin{tikzpicture}[>=stealth', shorten >=1pt, auto, node distance=2.50cm,initial text={}]
			\tikzset{sstate/.style={state,inner sep=3pt, minimum size=6pt}}
			\node [initial, accepting, sstate] (i) {$s$};
			\node [sstate, accepting] (s1) [right of=i] {$1$};
			\node (tmp1) [right of=s1,node distance=3.50cm]{};
			\node [sstate,node distance=2.00cm, accepting] (s0) [above of=tmp1] {$0$};
			\node [sstate,node distance=2.00cm, accepting] (s2) [below of=tmp1] {$2$};
			\path [->]
			      (i) edge [bend left=30] node [left] {$u[0]/u[0]\quad$} (s0)
			      (i) edge node [below] {$u[1]/u[1]$} (s1)
			      (i) edge [bend right=30] node [left] {$u[2]/u[2]\quad$} (s2)
			      (s1) edge [bend left=15] node [left] {$u[0]/u[0]\>$} (s0)
			      (s1) edge [bend right=15] node [left] {$u[2]/u[2]\>$} (s2)
			      (s0) edge [bend left=15] node [right] {$\>u[1]/u[1]$} (s1)
			      (s0) edge [bend left=35] node [left] {$\>u[2]/u[2]$} (s2)
			      (s2) edge [bend right=15] node [right] {$\>u[2]/u[2]$} (s1)
			      (s2) edge [bend right=65] node [right] {$\>u[0]/u[0]$} (s0)
			      (s0) edge [loop above] node [above] {$u[0]/\ews$} ()
			      (s1) edge [loop above] node [left] {$u[1]/\ews$} ()
			      (s2) edge [loop below] node [below] {$u[2]/\ews$} ()
			;
		\end{tikzpicture}
		\pmsn\parbox{0.85\textwidth}{\caption{Transducer realizing the command $\mathtt{tr}\>\mathtt{-s}\>\>u$ with $|u|=3$.}\label{FIG:unix}\prlabel{FIG:unix}}
	\end{center}\end{figure}
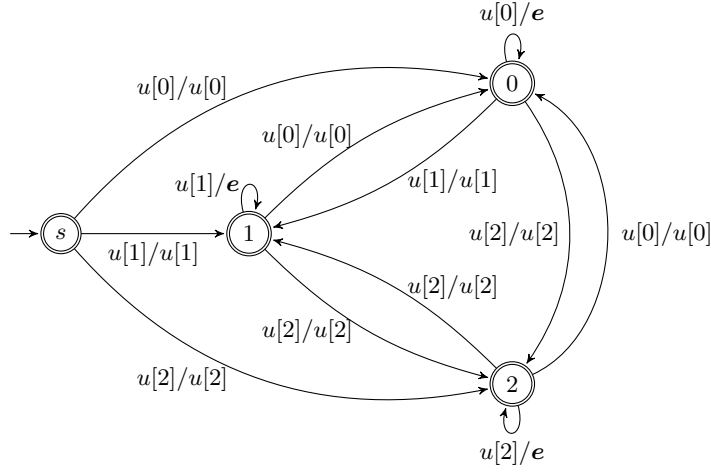
\end{example}

The Thompson method, \cite{Thom:1968}, of converting an ordinary regular expression over $\Sigma$---a type $\Sigma_{\ews}$ regular expression in the present terminology---to an $\ew$-NFA can be extended without complications to work with type $B$ regular expressions, for any label set $B$, using Lemma~\ref{LEM:rat:op}.

\begin{theorem}\label{TH:thompson}\prlabel{TH:thompson}
	Let $B$ be a label set with behaviour $\cali$. For each type $B$ regular expression $\rer$, there is a type $B$  graph $\autg(\rer)$ such that 
	\[\cali(\rer)=\cali\big(\autg(\rer)\big)\quad\text{ and }\quad \sz{\autg(\rer)}=O(|\rer|). \]
\end{theorem}

\pssi
For the converse of the above theorem, we shall extend the state elimination method of automata, \cite{BrzMcC:1963}, to labelled graphs. 
\pssi
Let $\autg=\big(Q,B,\delta,\{s\},\{f\}\big)$ be a type $\reg B$   graph, where $B$ is a label set with some behaviour $\cali$. We say that $\autg$ is \emdef{non-returning}, if  $s\neq f$, there are no transitions 
    going into $s$, there are no transitions 
    coming out of $f$, and there is at most one transition between any two states of $\autg$. For any states $p,r\in Q$, let $B_{p,r}=\{\beta\mid (p,\beta,r)\in\delta\}$, and let $\rer_{p,r}=\beta_1+\cdots+\beta_m$, where the $\beta_i$'s are the elements of $B_{p,r}$, if $B_{p,r}\neq\eset$. We define next the labelled graph $\auth$ that \emdef{results by eliminating} a state $q\in Q\setminus\{s,f\}$ from $\autg$. It is the type $\reg B$ graph
\begin{equation}\label{EQ:state:elim}\prlabel{EQ:state:elim}
	\auth =\big(Q\setminus\{q\},B,\delta',\{s\},\{f\}\big)
\end{equation}
	such that $\delta'$ is defined as follows. For any states $p,r\in Q\setminus\{q\}$:
\begin{itemize}
	\item If $(p,\alpha,r)\in\delta$ then $(p,\alpha,r)\in\delta'$.
	\item If $(p,\alpha_1,q),(q,\alpha_2,r)\in\delta$ then either $(p,\alpha_1(\rer_{q,q}^*)\alpha_2,r)\in\delta'$ if $B_{q,q}\neq\eset$, or $(p,\alpha_1\alpha_2,r)\in\delta'$ if $B_{q,q}=\eset$.
\end{itemize}

\begin{lemma}\label{LEM:state:elim}\prlabel{LEM:state:elim}
	Let $\autg=\big(Q,B,\delta,\{s\},\{f\}\big)$ be a non-returning  labelled graph, where $B$ is a label set with some behaviour $\cali$. If $\auth$ is the type $\reg B$   graph  that results by eliminating a state $q\in Q\setminus\{s,f\}$ from  $\autg$ then $\cali(\auth)=\cali(\autg)$.
\end{lemma}
\begin{proof} The main steps of the proof are analogous to those used in traditional proofs for the case of NFAs \cite{Wood:theory:of:comput}.
First, let $m\in\cali(\autg)$. Then, there is an accepting path $P=\sequ{q_{i-1},\beta_i,q_i}_{i=1}^\ell$ of $\autg$ such that $\ell\ge1$, $m=m_1\cdots m_\ell$ and each $m_i\in\cali(\beta_i)$. By a $q$-block of transitions in $P$ we mean a path $R=\sequ{q_{j-1},\beta_j,q_j}_{j=b}^{b+r}$ such that 
\pssi $r\ge1$, $q_{b-1}\neq q$, $q_b=\cdots=q_{b+r-1}=q$, $q_{b+r}\neq q$.
\pssn As  $\autg$ has at most one transition between any two states, we have that, if $r\ge2$, then $\beta_{b+1}=\cdots\beta_{b+r-1}=\rer_{q,q}$. Then,  
$$e=(q_{b-1},\beta_b(\rer_{q,q})^*\beta_{b+r},q_{b+r})\quad\text{or}\quad e=(q_{b-1},\beta_b\beta_{b+r},q_{b+r})$$
is a transition in $\delta'$---see~\eqref{EQ:state:elim}. Moreover, $m_b\cdots m_{b+r}\in\cali(\beta_b)\cali(\rer_{q,q}^*)\cali(\beta_{b+r})$. If we replace in $P$ the $q$-block $R$ with the transition $e$, and we repeat this with all $q$-blocks in $P$, then we get an accepting path of $\auth$ such that $m\in\cali(\auth)$.
\pnsi Conversely, let $m\in\cali(\auth)$. Then there is an accepting path $P'=\sequ{q_{i-1},\beta_i,q_i}_{i=1}^\ell$ of $\auth$ such that $\ell\ge1$, $m=m_1\cdots m_\ell$ and each $m_i\in\cali(\beta_i)$. For each $i\in\{1,\ldots,\ell\}$, we define a path $P_i$ of $\autg$ as follows. If $(q_{i-1},\beta_i,q_i)\in\delta$ then  $P_i=\sequ{q_{i-1},\beta_i,q_i}$. Else, there are $(q_{i-1},\alpha_1,q),(q,\alpha_2,q_i)\in\delta$ such that
$$\beta_i=\alpha_1(\rer_{q,q})^*\alpha_2\quad\text{or}\quad\beta_i=\alpha_1\alpha_2.$$
As $m_i\in\cali(\beta_i)$ and, if defined, $\cali(\rer_{q,q})^*=\cup_{r=0}^\infty\cali(\rer_{q,q})^r$, we have that there is $r\ge0$ such that
$$m_i=k_i'k_{i,1}\cdots k_{i,r}k_i'',\>\>k_i'\in\cali(\alpha_1),\>\>k_i''\in\cali(\alpha_2),\>\>k_{i,1},\ldots,k_{i,r}\in\cali(\rer_{q,q}).$$
Then, the path $P_i$ is
$$\sequ{(q_{i-1},\alpha_1,q),(q,\rer_{q,q},q),\ldots,(q,\rer_{q,q},q),(q,\alpha_2,q_i)},$$
which has $r$ repetitions of $(q,\rer_{q,q},q)$. Now define the sequence $P$ to be the concatenation of all paths $P_i$. This sequence is an accepting path of $\autg$ and this implies that $m\in\cali(\autg)$.
\end{proof}

As a type $B$  graph $\autg$ is also a type $\reg B$  graph, and as $\autg$ can be modified to be non-returning, we can apply the above lemma repeatedly until  we get a type $\reg B$ graph $\auth$ with set of states $\{s,f\}$ such that $\cali(\autg)=\cali(\auth)$. Then, we have that $\cali(\auth)=\cali(\rer_{s,f})$. Thus, we have the following consequence of Lemma~\ref{LEM:state:elim}.

\begin{corollary}\label{COR:state:elim}\prlabel{COR:state:elim}
	Let $B$ be a label set with behaviour $\cali$. For each type $B$  graph $\autg$ there is a type $B$ regular expression $\rer$ such that $\cali(\autg)=\cali(\rer)$.
\end{corollary}

\section{Partial Derivatives of Regular Expressions}\mylabel{SEC:derivatives}
\if\DRAFT1\textbf{\color{blue}
\pssi This section is incomplete and will be modified.
\pssi REMINDER: strings of set specs and pairing specs
}\fi 

Derivatives based methods for the manipulation of regular expressions have been widely studied~\cite{brzozowski64:_deriv_of_regul_expres,antimirov96:_partial_deriv_regul_expres_finit_autom_const,b.g.mirkin66:_algor_for_const_base_in,lombardy05:_deriv,broda11:_averag_state_compl_of_partial_deriv_autom,champarnaud02:_canon_deriv_partial_deriv_and,caron11:_partial_deriv_of_exten_regul_expres}. In recent years, partial derivative automata were defined and characterised for several kinds of expressions. Not only they are in general more succinct than other equivalent constructions but also for several operators they are easily defined (e.g. for intersection~\cite{Bastos:2017aa} or tuples~\cite{demaille16:_deriv_term_autom_multit_ration_expres}).
The partial derivative automaton of a regular expression  over $\Sigma^\star$ was introduced independently by Mirkin~\cite{b.g.mirkin66:_algor_for_const_base_in} and
Antimirov~\cite{antimirov96:_partial_deriv_regul_expres_finit_autom_const}.  
Champarnaud and Ziadi~\cite{champarnaud01:_from_mirkin_prebas_to_antim} proved that the two formulations are equivalent. Lombardy and Sakarovitch~\cite{lombardy05:_deriv} generalised these constructions to  weighted regular expressions, and recently Demaille~\cite{demaille16:_deriv_term_autom_multit_ration_expres} defined  derivative automata for  multitape weighted regular expressions.

Here we define the partial derivative automaton for regular expressions with set specifications. First however, we start with  more general regular expressions over  label sets.

Given a finite set  $S$ of expressions we define its behaviour as  $\cali(S)=\bigcup_{\res\in S}\cali(\res)$.
We say that  two regular expressions $\rer, \res$ of a type $B$ are  \emdef{equivalent}, $\rer\sim\res$,  if $\cali(\rer)=\cali(\res)$. Let the \emdef{set of label sets} of an expression $\rer$  be the set $\SoS(\rer)=\{\;\beta \mid \beta \in B \text{ and } \beta \text{ occurs in } \rer \;\}$. The \emdef{size} of  an expressions $\rer$ is   $\szw{\rer}=|\SoS(\rer)|$; it can be inductively   defined as follows:
 \begin{eqnarray*}
 \szw{\erel}&=&0\\
  \szw{\underline{\ew_{\mon B}}}&=&0\\
 	\szw{\beta}&=&1\\
 	\szw{\rer+\res}&=&\szw{\rer}+\szw{\res}\\
 	\szw{\rer\res}&=&\szw{\rer}+\szw{\res}\\
 	\szw{\rer^*}&=&\szw{\rer}.
 \end{eqnarray*}

Given a monoid $M$ the   \emdef{(left) quotient}  of  $P\subseteq M$ by an element $m\in M$ is defined by $\quo{m}{P}=\{\;m'\in M\mid mm'\in P\;\}$. This can be extended additively to the quotient of $P$ by another subset $Q$ of $M$, $\quo{Q}{P}$.

From now on we assume that $\>\mon B\>$ is either $\Sigma^*$ or $\Sigma^*\times \Gamma^*$, and that $\cali(\beta)$ is a subset of generators of $M$, for all $\beta\in B$.

 We define the \emdef{constant part} $\const: \reg B \to \{\underline{\ew_{\mon B}},\erel\}$ by
$\const(\rer)=\underline{\ew_{\mon B}}$ if $\ew_{\mon B}\in \cali(\rer)$, and
$\const(\rer)=\erel$ otherwise. This function is extended to sets of expressions by $\const(S)=\underline{\ew_{\mon B}}$ if and only if exists $\rer\in S$ such that $\const(\rer)=\underline{\ew_{\mon B}}$.

 The  \emdef{linear form}  of a regular expression   $\rer$,  $\lf(\rer)$,  is given by the following inductive definition:
\begin{eqnarray*}
 \lf(\erel) & =& \lf(\underline{\ew_{\mon B}}) = \emptyset,\\
\lf(\beta) & = & \{(\beta,\underline{\ew_{\mon B}})\},\\
\lf(\regexp\disj \regexp') &=&\lf(\regexp)\cup\lf(\regexp'),\\
\lf(\regexp\regexp') &=&
\begin{cases}
	\lf(\regexp)\regexp'\cup\lf(\regexp')&\text{if } \const(\regexp)=\underline{\ew_{\mon B}},\\
	\lf(\regexp)\regexp'&\text{otherwise},
\end{cases}\\
\lf(\regexp^*) &=&\lf(\regexp)\regexp^*,
\end{eqnarray*}
\noindent  where for any $S\subseteq B\times \reg B$, we define $S\underline{\ew_{\mon B}}=\underline{\ew_{\mon B}}S=S$, and $S\res=\{\,(\beta,\rer\res)\mid (\beta,\rer)\in S\,\}$ if $\res\neq\underline{\ew_{\mon B}}$ (and analogously for $\res S$).
Let $\cali(\lf(\rer))=\bigcup_{(\beta,\res)\in\lf(\rer)}\cali(\beta)\cali(\res)$.
\begin{lemma}
		For all $\rer\in \reg B$, $\rer\sim\const(\rer)\;\cup \;\lf(\rer).$
\end{lemma}

\begin{proof}
Proof by induction on $\rer$. \qed\end{proof}

For a regular expression $\rer$ and $\beta\in \SoS(\rer)$, the set of partial derivatives of $\rer$ w.r.t. $\beta$ is
$$\partial_\beta(\regexp)=\{\;\res \mid (\beta,\res)\in \lf(\rer)\;\}.$$ 

It is clear that we can iteratively compute the linear form of an expression  $\res\in \partial_\beta(\regexp)$, for $\beta \in \SoS(\rer)$. The set containing $\rer$ and  of all the resulting expressions is called   the set of partial derivatives of $\rer \in \reg B$, $\PD(\rer)$.

The \emdef{partial derivative graph} of  $\rer$ is
$$\apd(\rer)=\langle \PD(\rer),\Sigma,\delta_{\APD},\rer,F\rangle,$$
where
$F=\{\,\regexp_1\in\PD(\regexp)\mid \const(\regexp_1)=\underline{\ew_{\mon B}}\,\}$, and $\delta_{\APD}=\{\,(\regexp_1,\beta,\regexp_2) \mid \regexp_1\in\PD(\regexp)\,\land\, \beta\in\SoS[\rer] \land \regexp_2\in\partial_\beta(\regexp_1)\,\}$.

Whether this graph is finite and whether it is equivalent to $\rer$ (has the same behaviour, that is) depends on the behaviour $\cali$ of the label set $B$.

\subsection{Regular Expressions with Set Specifications}
\label{sec:resset}

Here we consider regular expressions of type $\sspec[\Sigma]$ which fixed behaviours are languages over alphabet $\Sigma$. Given $L_1, L_2\subseteq \Sigma^\star$ and $\letter \in \Sigma$, the quotient of a language  w.r.t $\letter$ satisfies the  following relations 
\begin{eqnarray*}
	\quo{\letter}{(L_1\cup L_2)}&=&\quo{\letter}{L_1}\cup \quo{\letter}{L_2},\\
	\quo{\letter}{(L_1\ L_2)}&=&\begin{cases}
	(\quo{\letter}{L_1})L_2&\text{ if} \varepsilon\notin L_1,\\
	(\quo{\letter}{L_1})L_2\cup \quo{\letter}{L_2}& \text{ if} \varepsilon\in L_1,
	\end{cases}\\
\quo{\letter}{L_1^\star}&=&(\quo{\letter}{L_1})L_1^\star.
\end{eqnarray*}

Quotients can be extended to words and languages: $\quo{\varepsilon}{L}=L$, $\quo{( w\letter)}{L}=\quo{\letter}{(\quo{w}{L})}$ and $\quo{L_1}{L}= \bigcup_{w\in L_1}\quo{w}{L}$.
%
If $L_1\subseteq L_2\subseteq \Sigma^\star$ then 
$\quo{L_1}{L}\subseteq \quo{L_2}{L}$ and $\quo{L}{L_1}\subseteq \quo{L}{L_2}$.

Given two set specifications  $\speca, \specb\in\sspec[\Sigma]\setminus\{\ews\}$  
we extend the notion of partial derivative to the set of partial derivatives of $\speca$ w.r.t $\specb$ with possible $\speca\not=\specb$, by
\begin{eqnarray*}
\partial_\speca(\specb) & = &\begin{cases}
\{\ews\} & \text{if }\speca\cap \specb\not=\bot,\\
\emptyset & \text{otherwise}.
\end{cases}
\end{eqnarray*}

Note that this definition is coherent with the definition of partial derivatives of an expression w.r.t. a label given before.
The set of partial derivatives of $\rer$ w.r.t.~a word $x \in
(B\setminus\{\underline{\ew_{\mon B}}\})^*$ is inductively defined by $\partial_\varepsilon(\rer)=\{
\rer\}$ and $ \partial_{x \beta  }(\rer) = \partial_\beta(\partial_x(\rer))$,
\noindent where, given a set $S \subseteq \reg B$, $\partial_\beta(S)= \bigcup_{\regexp
  \in S} \partial_{\beta}(\rer)$. Moreover one has
  $\lang(\partial_x(\rer))=\bigcup_{\rer_1\in\partial_x(\rer)}\lang(\rer_1)$.
 The following lemmas 
 ensure that the partial derivative automaton has the same behaviour of the correspondent regular expression.
 These results generalize known results for ordinary regular expressions~\cite{antimirov96:_partial_deriv_regul_expres_finit_autom_const,b.g.mirkin66:_algor_for_const_base_in}.

 \begin{lemma}\label{lem:quospec}
For two set specifications $\speca,\specb\in \sspec[\Sigma]$, 
$\quo{\lang(\speca)}
{\lang(\specb)}=\{\varepsilon\}$ if  $\speca\cap \specb\not= \bot$, and 
$\quo{\lang(\speca)}{\lang(\specb)}=\emptyset$ otherwise.	
\end{lemma}
\begin{proof} Given $L\subseteq \Sigma$ and $\letter \in \Sigma$ one has $\quo{\letter}{L}=\{\varepsilon\}$ if $\letter \in L$ and  $\quo{\letter}{L}=\emptyset$,  otherwise. Thus, the result follows.\qed
\end{proof}

 For instance, if $\sone w\cap \snone u\not=\bot$ then $$\quo{\lang(\sone w)}{\lang(\snone u)}=\bigcup_{\letter \in \al w}\quo{\letter}{(\Sigma\setminus \al u)}=\{\varepsilon\}.$$
\begin{lemma}\label{lem:lspecpartial}
For all $\regexp\in \REspec$ and $\speca\in \sspec[\Sigma]$, 
$\quo{\lang(\speca)}{\lang(\regexp)}=\lang(\partial_\speca(\regexp))$.
\end{lemma}
\begin{proof} For $\regexp=\emptyset$ and  $\regexp=\ews$ it is obvious. For  $\regexp=\specb$ the result  follows from Lemma~\ref{lem:quospec}. In fact, if 
$\quo{\lang(\speca)}{\lang(\specb)}=\{\varepsilon\}$
then  $\partial_\speca(\specb)=\{\ews\}$ and thus $\lang(\partial_\speca(\regexp))= \{\varepsilon\}$; otherwise if $\quo{\lang(\speca)}{\lang(\specb)}=\emptyset$ then $\partial_\speca(\specb)=\emptyset$, and also $\lang(\partial_\speca(\regexp))=\emptyset$.  The remaining cases follow by induction as with ordinary regular expressions.
\qed
\end{proof}

\begin{lemma}\label{lem:wordspecartial}
For all $g\in(\sspec[\Sigma]\setminus\{\ews\})^\star$,
$\quo{\lang(g)}{\lang(\regexp)}=\lang(\partial_g(\regexp))$.
\end{lemma}
\begin{proof}
By induction on $|g|$ using Lemma~\ref{lem:lspecpartial}.	\qed
\end{proof}

\begin{lemma}\label{LEM:regpartial}\prlabel{LEM:regpartial}
	For all $w\in \Sigma^\star$, the following propositions are equivalent:
	\begin{enumerate}
	\item 	$w\in \lang(\regexp)$\label{prop:rpin}
	\item  $w=\letter_1\cdots\letter_n$ and there exists $\ws(w)=\speca_1\cdots \speca_n$ with $\speca_i\in \SoS(\regexp)$, $\sone \letter_i\cap \speca_i\not=\bot$ and  $\const(\partial_{\ws(w)}(\regexp))=\varepsilon$\label{prop:rpcap}.
	\end{enumerate}

\end{lemma}
\begin{proof}
	For $w=\varepsilon$, $n=0$ and $\const(\regexp)=\varepsilon$ if and only if
	$\const(\partial_{\varepsilon}(\regexp))=\varepsilon$. 
	For $w\not=\varepsilon$ we prove by induction on the structure of $\regexp$.

If $\regexp=\speca\not=\ews$  and $w\in \lang(\regexp)$,  $w$ is some letter $\letter$. 
 Then $\sone \letter\cap \speca\not=\bot$ and also $\const(\partial_{\speca}(\regexp))=\varepsilon$. If \ref{prop:rpcap}. holds then $\partial_{\ws(x)}(\regexp)\not=\emptyset$ if and only if $\ws(w)=\speca$. Because  $\sone \letter\cap \speca\not=\bot$ we have $w\in \lang(\regexp)$.
 	
 	Suppose the result holds for $\regexp_1$ and $\regexp_2$.
 	
 	Let $\regexp=\regexp_1\disj\regexp_2$. If $w\in \lang(\regexp)$ suppose without lost of generality that $w\in \lang(\regexp_1)$. By the induction hypothesis
 	there exists 
 	$\ws(w)=\speca_1\cdots \speca_n$ with $\speca_i\in \SoS(\regexp_1)$, $\sone \letter_i\cap \speca_i\not=\bot$ and  $\const(\partial_{\ws(x)}(\regexp_1))=\varepsilon$.
 Using Lemma~\ref{lem:wordspecartial},  
  $\varepsilon \in \quo{\lang(\ws(w)}{\lang(\regexp_1)}\subseteq \quo{\lang(\ws(w))}{\lang(\regexp_1+\regexp_2)}=\quo{\lang(\ws(w))}{\lang(\regexp)}$ and thus	 $\const(\partial_{\ws(w)}(\regexp))=\varepsilon$.
  If \ref{prop:rpcap}. holds suppose without lost of generality that $\varepsilon\in \lang(\partial_{\ws(w)}(\regexp_1))$, i.e. $\const(\partial_{\ws(w)}(\regexp_1))=\varepsilon$. Then $w\in \lang(\regexp_1)\subseteq  \lang(\regexp)$.
  
 		For $\regexp=\regexp_1\concat\regexp_2$ and $w\in \lang(\regexp)$, $w=w_1w_2$ with $w_i\in  
 		\lang(\regexp_i)$,  for $i=1,2$.
 		If  $w_1=\varepsilon$ let $w_2=\letter_1\cdots \letter_n$.  By the induction hypothesis
 	there exists 
 	${\ws(w)}=\speca_1\cdots \speca_n$ with $\speca_i\in \SoS(\regexp_2)$, $\sone \letter_i\cap \speca_i\not=\bot$ and  $\const(\partial_{\ws(w)}(\regexp_2))=\varepsilon$.
 Using Lemma~\ref{lem:wordspecartial},  
  $\varepsilon \in \quo{\lang(\ws(w))}{\lang(\regexp_2)}\subseteq \quo{\lang(\ws(w))}{\lang(\regexp_1\concat\regexp_2)}=\quo{\lang(\ws(w))}{\lang(\regexp)}$ and thus	 $\const(\partial_{\ws(w)}(\regexp))=\varepsilon$.  
  If $w_1=\letter_1'\cdots\letter_m'\not=\varepsilon$ then there also exists $\ws(w')=\speca_1'\cdots \speca_m'$ with $\speca_i'\in \SoS(\regexp_1)$, $\sone \letter_i'\cap \speca_i'\not=\bot$ and  $\const(\partial_{w_{x}'}(\regexp_1))=\varepsilon$. 
If $w_2=\varepsilon$ then  $\varepsilon \in \quo{\lang(\ws(w'))}{\lang(\regexp_1)}\subseteq\quo{\lang(\ws(w'))}{\lang(\regexp_1\concat\regexp_2)}=\quo{\lang(\ws(w'))}{\lang(\regexp)}$ and thus	 $\const(\partial_{\ws(w')}(\regexp))=\varepsilon$. Otherwise, let $\ws(w)$ be as in the case of $w_1=\varepsilon$, and one concludes that $\varepsilon\in \quo{\lang(\ws(w')\ws(w))}{\lang(\regexp)}=\quo{\lang(\ws(w))}{(\quo{\lang(\ws(w'))}{\lang(\regexp_1)})\lang(\regexp_2)}$ and thus	 $\const(\partial_{\ws(w')\ws(w)}(\regexp))=\varepsilon$.

If ~\ref{prop:rpcap}. holds then $\varepsilon\in \quo{\lang(\ws(w)}{(\lang(\rer_1)\lang(\regexp_2))}=\quo{(\lang(\speca_1)\cdots\lang(\speca_n))}{(\lang(\rer_1)\lang(\regexp_2))}$. We have three cases to consider:
\begin{enumerate}[a)]
	\item $\const(\partial_{\ws(w)}(\regexp_1))=\varepsilon$
	\item $\const(\regexp_1)=\varepsilon$ and $\const(\partial_{\ws(w)}(\regexp_2))=\varepsilon$
	\item $\ws(w)=\ws(u)\ws(v)$ with $\ws(u)=B_1\cdots B_j$ and $\ws(v)=B_{j+1}\cdots B_n$ and 
	$\const(\partial_{\ws(u)}(\regexp_1))=\varepsilon$ and 	$\const(\partial_{\ws(v)}(\regexp_2)=\varepsilon$.
\end{enumerate}

For $\regexp=\regexp_1^\star$, if $w\in \lang(\regexp)$ there exists $n$ such that $w\in \lang(\regexp_1)^n$. Then the proof is similar to the case of concatenation.\qed
 	\end{proof}

The set of all partial derivatives of~$\rer$ w.r.t.\ non-empty words is denoted by $\partial^+(\rer)$. Then $\PD(\rer)=\partial^+(\rer)\cup \{\rer\}$.


\begin{lemma}
\label{lemma:pdpcontainseps}
For $\regexp \in \REspec$, the following hold.
\begin{enumerate}
\item If $\partial^+(\regexp) \not= \emptyset$, then there is $\regexp_0 \in \partial^+(\regexp)$ with $\const(\regexp_0) = \varepsilon$.
\item If $\partial^+(\regexp) = \emptyset$ and $\regexp \not=\erel$, then $\lang(\regexp) = \{ \varepsilon \}$ and $\const(\regexp)=\varepsilon$.
\end{enumerate}
\end{lemma}
\begin{proof}
  \begin{enumerate}
  \item From the definition of regular expressions follows that
    $\erel$ cannot appear as a subexpression of a larger term.
    Suppose that there is some $\regexpc \in \partial^+(\regexp)$ and $\lang(\regexpc)\not= \emptyset$. Then there is some word $w \in \Sigma^*$ such    
    that 
     $w \in \lang(\regexpc)$. Then $\varepsilon \in
    \lang(\partial_{\ws(w)}(\regexpc))$ which means that there is some
    $\regexp_0 \in \partial_{\ws(w)}(\regexpc) \subseteq \partial^+(\regexp)$ such
    that $\const(\regexp_0)=\varepsilon.$
  \item $\partial^+(\regexp) = \emptyset$ implies that there is no word $z \in \Sigma^+$ in
    $\lang(\regexp)$. On the other hand, since $\erel$ does not
    appear in $\regexp$, it follows that $\lang(\regexp) \not=
    \emptyset$. Thus, $\lang(\regexp) = \{ \varepsilon \}$.\qed
\end{enumerate}
\end{proof}

The following proposition generalizes from ordinary regular expressions~\cite{b.g.mirkin66:_algor_for_const_base_in,champarnaud01:_from_mirkin_prebas_to_antim,broda11:_averag_state_compl_of_partial_deriv_autom}, and shows that the set of partial derivatives is finite.

\begin{lemma}\label{LEM:pdssp_recursion}\prlabel{LEM:pdssp\_recursion}
  $\partial^+$ satisfies the following:\\
\begin{equation*}
  \begin{array}[t]{rcl}
    \partial^+(\emptyset)&=&\partial^+(\ews) = \emptyset, \\
    \partial^+(\speca) &=& \{ \varepsilon\} \\
     \partial^+(\regexp^\star) &=& \partial^+(\regexp)\regexp^\star.
      \end{array}
      \hspace{1cm}
      \begin{array}[t]{rcl}
      \partial^+(\regexp_1+ \regexp_2) &=& \partial^+(\regexp_1) \cup \partial^+(\regexp_2),\\
      \partial^+(\regexp_1 \regexp_2) &=& \partial^+(\regexp_1) \regexp_2
      \cup \partial^+(\regexp_2),
    \end{array}
  \end{equation*}
\end{lemma}

\begin{theorem}
\begin{eqnarray*}
|\partial^+(\regexp)|&\leq& \szw{\regexp}, \\
	|\PD(\regexp)|&\leq& \szw{\regexp} +1.
\end{eqnarray*}
\end{theorem}

\begin{proof}
	Direct consequence of  Lemma~\ref{LEM:pdssp_recursion} using induction on $\regexp$.\qed
\end{proof}

The following proposition shows that
the \emph{partial derivative automaton} of~$\rer$ with set specifications  is equivalent to $\rer$.

\begin{theorem}\label{TH:apd}\prlabel{TH:apd}
$\lang(\apd(\regexp))=\lang(\regexp)$.
\end{theorem}

\begin{proof}
	By induction on $|g|$ with $g\in \SoS(\regexp)^\star$,  one can prove that   there is a path from $\regexp_1$ to $\regexp_2$ labeled by $g$ if and only if $\regexp_2\in \partial_{g}(\regexp_1)$, for any $\regexp_1\in \PD(\regexp)$. Now, we prove that for any $\regexp_1\in \PD(\regexp)$ and $w\in \Sigma^\star$,
	$$w\in \lang_{\regexp_1}(\apd(\regexp))\; \Leftrightarrow \; w\in \lang(\regexp_1).$$
	Let $w=\letter_1\cdots\letter_n$.
	 If $w\in \lang(\regexp_1)$ applying Lemma~\ref{LEM:regpartial} 
	 one concludes that
	 there exists $\ws(w)=\speca_1\cdots \speca_n$ 
	such that $\const(\partial_{\ws(w)}(\regexp_1))=\varepsilon$. Then there exists $\regexp_2\in \partial_{\ws(w)}(\regexp_1)$ such that $\const(\regexp_2)=\varepsilon$, and thus $w\in \lang_{\regexp_1}(\apd(\regexp))$.
	
	 If $w\in \lang_{\regexp_1}(\apd(\regexp))$, there is an accepting path from $\regexp_1$ to a state $\regexp_2$ labeled by $\ws(w)=\speca_1\ldots \speca_n$, $\const(\regexp_2)=\varepsilon$ and $w\in \lang(\speca_1)\cdots \lang(\speca_n)$. Then we conclude that $\sone \letter_i\cap \speca_i\not=\bot$ for $i=1,\ldots n$, and again by Lemma~\ref{LEM:regpartial}, $w\in  \lang(\regexp_1)$.\qed
\end{proof}

\if\DRAFT1
%
\subsection{Regular Expressions with Pairing Specifications}
\label{sec:pdpairspec}
Consider the monoid $\Sigma^*\times\Delta^*$ with the following set of generators $\{(\letter,\varepsilon), (\varepsilon,\letter')\mid \letter\in\Sigma \land \letter'\in \Delta\}$ and set of equations 
$$\{\, (\letter,\varepsilon)(\varepsilon,\letter')\eqmod (\letter,\letter'), (\varepsilon,\letter')(\letter,\varepsilon)\eqmod (\letter,\letter') \mid  \letter\in\Sigma \land \letter'\in \Delta\,\}.$$
 For $R\subseteq \Sigma^*\times\Delta^*,$  the quotient of $R$ by a symbol is defined  as before but one needs to take in account the above equations.
 For instance, for $\letter\in \Sigma$ and $\letter'\in \Delta$:
\begin{eqnarray*}
	\quo{\pair{\letter}{\varepsilon}}R &=&\{\,(\varepsilon,\lettera)w\mid (\letter,\lettera)w\in R\},\\
\quo{\pair{\varepsilon}{\letter'}}R &=&\{\,(\letter,\varepsilon)w\mid (\letter,\letter')w\in R \}.
\end{eqnarray*}

Partial derivatives of a  regular expression  $\rer\in \reg [\Sigma_{\ews},\Delta_{\ews}]$ can be defined as for ordinary regular expressions, considering for $\tuple{x}{y} \in [\Sigma_{\ews},\Delta_{\ews}]$  the following:
\begin{eqnarray*}
	\partial_{\tuple{\letter}{\lettera}}(\tuple{\letter'}{\lettera'}) & = &\begin{cases}
\{\tuple{\ews}{\ews}\} & \text{if } \letter=\letter'\land \lettera=\lettera',\\
\{\tuple{\ews}{\lettera'}\} & \text{if } \letter=\letter'\land \lettera=\ews,\\
\{\tuple{\letter'}{\ews}\} & \text{if } \lettera=\lettera'\land \letter=\ews,\\
\emptyset & \text{otherwise}.
\end{cases}
\end{eqnarray*}
With the above definitions the results of Section~\ref{sec:resset} (and the usual ones for ordinary regular expressions) apply to regular expressions of type $[\Sigma_{\ews},\Delta_{\ews}]$. But, for every $w\in \Sigma^*\times\Delta^*$ it is needed to consider all it possible representations.
This means that one can define 
\begin{equation}\label{eq:pdwordsprod}
\Delta_{\tuple{w}{w'}}\rer=\bigcup_{w=\letter w_1}\partial_{\tuple{w_1}{w'}}(\partial_{\tuple{\letter}{\varepsilon}}\rer )\cup \bigcup_{w'=\lettera w_2}\partial_{\tuple{w}{w_2}}(\partial_{\tuple{\varepsilon}{\lettera}}\rer). 
\end{equation}
For  regular expressions with pairing specifications it is necessary to define the set of partial derivatives of a pairing specification w.r.t other pairing  specification. For $\spp, \spp'\in \srel[\Gamma]\setminus\{\tuple{\ews}{\ews}\}$ and $\rel{\spp}\not=\emptyset$ the set of partial derivatives w.r.t $\spp$ of a pairing spec $\spp'$ is inductively defined below. For each case, if the conditions  do not hold then the value of $\partial_{\spp}(\spp')$ is $\emptyset$.

\begin{eqnarray*}
	\partial_{\tuple{\ews}{\speca}}(\tuple{\ews}{\specb})&=&\{\tuple{\ews}{\ews}\}\;\;\text{ if } \speca\cap \specb\not=\bot,\\
	 \partial_{\tuple{\ews}{\speca}}(\tuple{\specb}{\ews})&=&\emptyset,\\
	 \partial_{\tuple{\ews}{\speca}}(\tuple{\specb}{\specc})&=&\{\tuple{\specb}{\ews}\}\;\;\text{ if } \speca\cap \specc\not=\bot,\\
	 \partial_{\tuple{\ews}{\speca}}(\ssame{\specb})&=&\{\tuple{\speca\cap\specb}{\ews}\}\;\; \text{ if } \speca\cap \specb\not=\bot,\\
	  \partial_{\tuple{\ews}{\speca}}(\sdiff{\specb}{\specc})&=&\begin{cases}
\{\tuple{\specb\cap \snone b}{\ews}\} & \text{ if } \lang(\speca\cap\specc)=\{b\} \land \lang(\specb)\setminus\{b\}\not=\emptyset,\\
\{\tuple{\specb}{\ews}\} & \text{ if } |\lang(\speca\cap\specc)| 	\geq 2,
\end{cases}\\
\partial_{\tuple{\speca}{\ews}}(\tuple{\ews}{\specb})&=&\emptyset,\\
	 \partial_{\tuple{\speca}{\ews}}(\tuple{\specb}{\ews})&=&\{\tuple{\ews}{\ews}\}\;\;\text{ if } \speca\cap \specb\not=\bot,\\
	 \partial_{\tuple{\speca}{\ews}}(\tuple{\specb}{\specc})&=&\{\tuple{\ews}{\specc}\}\;\;\text{ if } \speca\cap \specb\not=\bot,\\
	 \partial_{\tuple{\speca}{\ews}}(\ssame{\specb})&=&\{\tuple{\ews}{\speca\cap\specb}\}\;\; \text{ if } \speca\cap \specb\not=\bot,\\
	  \partial_{\tuple{\speca}{\ews}}(\sdiff{\specb}{\specc})&=&\begin{cases}
\{\tuple{\ews}{\specc\cap \snone b}\} & \text{ if } \lang(\speca\cap\specb)=\{b\} \land \lang(\specc)\setminus\{b\}\not=\emptyset,\\
\{\tuple{\ews}{\specc}\} & \text{ if } |\lang(\speca\cap\specb)| 	\geq 2,
\end{cases}\\
\partial_{\tuple{\speca}{\specd}}(\tuple{\ews}{\specb})&=&\emptyset,\\
	 \partial_{\tuple{\speca}{\specd}}(\tuple{\specb}{\ews})&=&\emptyset\\
	 \partial_{\tuple{\speca}{\specd}}(\tuple{\specb}{\specc})&=&\{\tuple{\ews}{\ews}\}\;\;\text{ if } \speca\cap \specb\not=\bot \land \specd\cap \specc\not=\bot  ,\\
	 \partial_{\tuple{\speca}{\specd}}(\ssame{\specb})&=&\{\tuple{\ews}{\ews}\}\;\; \text{ if } \speca\cap \specb\cap \specd\not=\bot,\\
	  \partial_{\tuple{\speca}{\specd}}(\sdiff{\specb}{\specc})&=&
\{\tuple{\ews}{\ews}\} \text{ if } 
\speca\cap \specb\not=\bot \land \specd\cap \specc \not=\bot\\
&& \text{ and if } 
\lang(\speca\cap\specb)=\{b\} \text{ then } \lang(\specd\cap \specc)\setminus\{b\}\not=\emptyset, \\
&& \text{ and if } \lang(\specd\cap\specc)=\{b\} \text{ then } \lang(\speca\cap \specb)\setminus\{b\}\not=\emptyset,\\
\partial_{\ssame{\speca}}(\tuple{\ews}{\specb})&=&\emptyset,\\
	 \partial_{\ssame{\speca}}(\tuple{\specb}{\ews})&=&\emptyset\\
	 \partial_{\ssame{\speca}}(\tuple{\specb}{\specc})&=&\{\tuple{\ews}{\ews}\}\;\;\text{ if } \speca\cap \specb\cap \specc\not=\bot, \\
	 \partial_{\ssame{\speca}}(\ssame{\specb})&=&\{\tuple{\ews}{\ews}\}\;\;\text{ if } \speca\cap \specb\not=\bot \\
	  \partial_{\ssame{\speca}}(\sdiff{\specb}{\specc})&=&
\emptyset,\\
\partial_{\sdiff{\speca}{\specd}}(\tuple{\ews}{\specb})&=&\emptyset,\\
	 \partial_{\sdiff{\speca}{\specd}}(\tuple{\specb}{\ews})&=&\emptyset\\
	 \partial_{\sdiff{\speca}{\specd}}(\tuple{\specb}{\specc})&=&\{\tuple{\ews}{\ews}\}\;\;\text{ if } 
\speca\cap \specb\not=\bot \land \specd\cap \specc \not=\bot\\
&& \text{ and if } 
\lang(\speca\cap\specb)=\{b\} \text{ then } \lang(\specd\cap \specc)\setminus\{b\}\not=\emptyset, \\
&& \text{ and if } \lang(\specd\cap\specc)=\{b\} \text{ then } \lang(\speca\cap \specb)\setminus\{b\}\not=\emptyset,\\
	 \partial_{\sdiff{\speca}{\specd}}(\ssame{\specb})&=&\emptyset,\\
	  \partial_{\sdiff{\speca}{\specd}}(\sdiff{\specb}{\specc})&=&
\{\tuple{\ews}{\ews}\} \text{ if } 
\speca\cap \specb\not=\bot \land \specd\cap \specc \not=\bot\\
&& \text{ and if } 
\lang(\speca\cap\specb)=\{b\} \text{ then } \lang(\specd\cap \specc)\setminus\{b\}\not=\emptyset, \\
&& \text{ and if } \lang(\specd\cap\specc)=\{b\} \text{ then } \lang(\speca\cap \specb)\setminus\{b\}\not=\emptyset. \\
\end{eqnarray*}

Using the  behaviours of each pairing specification we have:
\begin{lemma}
	For $\spp, \spp'\in \srel[\Gamma]\setminus\{\tuple{\ews}{\ews}\}$ and $\rel{\spp}\not=\emptyset$, $\rel{\partial_\spp(\spp')}=\quo{\rel{\spp}}{\rel{\spp'}}$.
\end{lemma}

Considering Equation~\eqref{eq:pdwordsprod} and the methods of Section~\ref{sec:resset} one can proof that the partial derivative construction for regular expressions with pairing specifications is correct, i.e., the transducer is equivalent to the expression.
\fi  

\section{Label Operations and the Product Construction}\mylabel{SEC:product}
We shall consider partial operations $\lop$ on label sets $B,B'$ such that, when defined, the product $\beta\odot\beta'$ of two labels belongs to a certain label set $C$. Moreover, we shall assume that $\lop$ is also a partial operation on $\mon B,\mon B'$ such that, when defined, the product $m\lop m'$ of two monoid elements belongs to $\mon C$. We shall call $\lop$ a \emdef{polymorphic} operation (in analogy to polymorphic operations in programming languages) when $\cali(\beta\odot\beta')=\cali_1(\beta)\lop\cali_1(\beta')$ where $\cali_1,\cali_2,\cali$ are the behaviours of $B,B',C$. This concept shall allow us to also use $\lop$ as the name of the  product construction on labelled graphs that respects the behaviours of the two graphs.

\begin{example}\label{EX:monoidops}\prlabel{EX:monoidops}
	We shall consider the following monoidal operations, which are better known when applied to subsets of the monoid.
	\begin{itemize}
    \item $\cap:\Sigma^*\times\Sigma^*\pto\Sigma^*$ such that $u\cap v=u$ if $u=v$; else, $u\cap v=\undef$. Of course, for any two languages $K,L\sseq\Sigma^*$, $K\cap L$ is the usual intersection of $K,L$.
    \item $\circ:(\Sigma_1^*\times\Delta^*)\times(\Delta^*\times\Sigma_2^*)\pto(\Sigma_1^*\times\Sigma_2^*)$ such that $(u,v)\circ (w,z)=(u,z)$ if $v=w$; else, $(u,v)\circ (w,z)=\undef$. For any two relations $R,S$, $R\circ S$ is the usual composition of $R,S$.
    \item $\rinp:(\Sigma^*\times\Delta^*)\times\Sigma^*\pto(\Sigma^*\times\Delta^*)$ such that $(u,v)\rinp w=(u,v)$ if $u=w$; else, $(u,v)\rinp w=\undef$. For a relation $R$ and language $L$, 
    \begin{equation}\label{EQ:rinp}\prlabel{EQ:rinp}
    	R\rinp L=R\cap (L\times\Delta^*).
    \end{equation} 
    \item $\rout:(\Sigma^*\times\Delta^*)\times\Sigma^*\pto(\Sigma^*\times\Delta^*)$ such that $(u,v)\rout w=(u,v)$ if $v=w$; else, $(u,v)\rinp w=\undef$. For a relation $R$ and language $L$, 
    \begin{equation}\label{EQ:rout}\prlabel{EQ:rout}
    	R\rout L=R\cap (\Sigma^*\times L).
    \end{equation} 
    \end{itemize}
\end{example}

\begin{definition}\label{DEF:polymorphic}\prlabel{DEF:polymorphic}
	Let $B,B',C$ be label sets with behaviours $\cali_1,\cali_2,\cali$, respectively. A \emdef{polymorphic operation} $\lop$ over $B,B',C$, denoted as ``$\odot:B\times B'\Rightarrow C$'', is defined as follows.
	\begin{itemize}
    \item It is a partial mapping:\qquad   $\odot:B\times B'\pto C$.
    \item It is a partial mapping:\qquad   $\odot:\mon B\times\mon B'\pto\mon C$.
    \item For all $\beta\in B$ and $\beta'\in B'$ we have
    \[
    \cali(\beta\lop\beta')=\cali_1(\beta)\lop\cali_2(\beta'),
    \]
    where we assume that $\cali(\beta\lop\beta')=\eset$, if $\beta\lop\beta'=\undef$; and we have used the notation
    \[
      S\lop S' = \{m\lop m'\mid m\in S, m'\in S', m\lop m'\neq\undef\}.
    \]
    for any $S\sseq \mon B$ and $S'\sseq \mon B'$.
    \end{itemize}
\end{definition}

\begin{example}\label{EX:polymor:old}\prlabel{EX:polymor:old}
	The following polymorphic operations are based on label sets of standard automata and transducers using the monoidal operations in Ex.~\ref{EX:monoidops}.
	\begin{itemize}
    \item ``$\cap:\Sigma_{\ews}\times\Sigma_{\ews}\Rightarrow\Sigma_{\ews}$'' is defined by 
    \begin{itemize}
        \item the partial operation $\cap: \Sigma_{\ews}\times\Sigma_{\ews}\pto\Sigma_{\ews}$ such that $x\cap y=x$, if $x=y$, else $x\cap y=\undef$; and 
        \item the partial operation $\cap:\Sigma^*\times\Sigma^*\pto\Sigma^*$. 
    \end{itemize}
    Obviously, $\lof{x\cap y}=\lof x\cap\lof y$.
    \item ``$\circ:[\Sigma_{\ews},\Delta_{\ews}]\times[\Delta_{\ews},\Sigma'_{\ews}]\Rightarrow[\Sigma_{\ews},\Sigma'_{\ews}]$'' is defined by 
        \begin{itemize}
        \item the operation $\circ:[\Sigma_{\ews},\Delta_{\ews}]\times[\Delta_{\ews},\Sigma'_{\ews}]\pto[\Sigma_{\ews},\Sigma'_{\ews}]$ such that $(x/y_1)\circ(y_2/z)=(x/z)$ if $y_1=y_2$, else $(x/y_1)\circ(y_2/z)=\undef$; and 
        \item the operation $\circ:(\Sigma^*\times\Delta^*)\times(\Delta^*\times\Sigma'^*)\pto(\Sigma^*\times\Sigma'^*)$.
        \end{itemize} 
        Obviously, $\rel{(x,y_1)\circ (y_2,z)}=\rel{(x,y_1)}\circ\rel{(y_2,z)}$.
    \item ``$\rinp:[\Sigma_{\ews},\Delta_{\ews}]\times\Sigma_{\ews}\Rightarrow[\Sigma_{\ews},\Delta_{\ews}]$'' is defined by 
        \begin{itemize}
        \item the operation $\rinp:[\Sigma_{\ews},\Delta_{\ews}]\times\Sigma_{\ews}\pto[\Sigma_{\ews},\Delta_{\ews}]$ such that $(x/y)\rinp z=(x/y)$ if $x=z$, else $(x/y)\rinp z=\undef$; and 
        \item the operation $\rinp:(\Sigma^*\times\Delta^*)\times \Sigma^*\pto(\Sigma^*\times\Delta^*)$. 
        \end{itemize} 
        Obviously, $\rel{(x/y)\rinp z}=\rel{x/y}\rinp\lof{z}$.
    \item ``$\rout:[\Sigma_{\ews},\Delta_{\ews}]\times\Delta_{\ews}\Rightarrow[\Sigma_{\ews},\Delta_{\ews}]$'' is defined by 
        \begin{itemize}
        \item the operation $\rout:[\Sigma_{\ews},\Delta_{\ews}]\times\Delta_{\ews}\pto[\Sigma_{\ews},\Delta_{\ews}]$ such that $(x/y)\rout z=(x/y)$ if $x=z$, else $(x/y)\rout z=\undef$; and 
        \item the operation $\rout:(\Sigma^*\times\Delta^*)\times \Sigma^*\pto(\Sigma^*\times\Delta^*)$. 
        \end{itemize} 
        Obviously, $\rel{(x/y)\rout z}=\rel{x/y}\rout\lof{z}$.
    \end{itemize}
\end{example}

\begin{example}\label{EX:polymor:new}\prlabel{EX:polymor:new}
    The following polymorphic operations are based on label sets of  automata and transducers with set specs.
    \begin{itemize}
    \item ``$\cap:\sspec[\Gamma]\times\sspec[\Gamma]\Rightarrow\sspec[\Gamma]$'' is defined by the partial operation
    $\cap:\sspec[\Gamma]\times\sspec[\Gamma]\pto\sspec[\Gamma]$, according to Def.~\ref{DEF:ssets}, and the partial operation $\cap:\Gamma^*\times\Gamma^*\pto\Gamma^*$. By Lemma~\ref{LEM:ssets:lang}, for any $B,F\in\sspec[\Gamma]$, we have that $$\lof{B\cap F}=\lof B\cap \lof F.$$
    \item ``$\rinp:\srel[\Gamma]\times\Gamma_{\ews}\Rightarrow\srel[\Gamma]$'' is defined as follows. First, by the partial operation $\rinp:\srel[\Gamma]\times\Gamma_{\ews}\pto\srel[\Gamma]$ such that 
\[
\spp\rinp x=\begin{cases}
	\ews/\spright\spp,& \text{if $x=\ews$ and $\spleft\spp=\ews$;} \\
	\sone x/\spright\spp,& \text{if $x,\spleft\spp\neq\ews$ and $x\in\lof{\spleft\spp}$}; \\
	\bot, & \text{otherwise. } 
\end{cases}
\]
    Second, by the partial operation $\rinp:(\Sigma^*\times\Delta^*)\times\Sigma^*\pto(\Sigma^*\times\Delta^*)$. We have that
    \[\rel{\spp\rinp x}=\rel{\spp}\rinp\lof x \]
    Moreover we have that $\spp\rinp x$ can be computed from $\spp$ and $x$ in time $O(\szabc\spp)$.
    \item ``$\rout:\srel[\Gamma]\times\Delta_{\ews}\Rightarrow\srel[\Gamma]$'' is defined as follows. First, by the partial operation $\rout:\srel[\Gamma]\times\Delta_{\ews}\pto\srel[\Gamma]$ such that $\spp\rout x=(\spp^{-1}\rinp x)^{-1}$. Second, by the partial operation $\rout:(\Sigma^*\times\Delta^*)\times\Delta^*\pto(\Sigma^*\times\Delta^*)$. We have that
        \[\rel{\spp\rout x}=\rel{\spp}\rout\lof x \]
        Moreover we have that $\spp\rout x$ can be computed from $\spp$ and $x$ in time $O(\szabc\spp)$.
   \end{itemize}
   Further below, in Sect.~\ref{SEC:compose}, we define the polymorphic operation `$\circ$' between pairing specs.
\end{example}

\begin{definition}\label{DEF:product}\prlabel{DEF:product}
	Let $\autg=(Q,B,\delta,I,F)$ and $\autg'=(Q',B',\delta',I',F')$ be type $B$ and $B'$, respectively,    graphs and let ``$\lop:B\times B'\Rightarrow C$''  be a polymorphic operation. The \emdef{product} $\autg\lop\autg'$ is the type $C$   graph 
	$$\big(P,C,\delta\lop\delta',I\times I',F\times F'\big)$$ 
	defined as follows. First make the following two possible modifications on $\autg,\autg'$: if there is a label $\beta$ in $\autg$ such that $\ew_{\mon B}\in\cali(\beta)$ then modify $\autg'$ to $\autg'^{\ew}$; and if there is a label $\beta'$ in $\autg'$ (before being modified) such that $\ew_{\mon B'}\in\cali(\beta')$ then modify $\autg'$ to $\autg'^{\ew}$. In any case, use the same names $\autg$ and $\autg'$ independently of whether they were modified. Then $P$ and $\delta\lop\delta'$ are defined inductively as follows: 
	\begin{enumerate}
		\item $I\times I'\subseteq P$.
		\item If $(p,p')\in P$ and there are $(p,\beta,q)\in\delta$ and $(p',\beta',q')\in\delta'$ with $\beta\lop \beta'\neq\undef$ then $(q,q')\in P$ and $\big((p,p'),\beta\lop\beta',(q,q')\big)\in\delta\lop\delta'$. 
	\end{enumerate}
\end{definition}

\begin{example}\label{EX:product}\prlabel{EX:product}
	Here we recall three known examples of product constructions involving automata and transducers. 
	\begin{enumerate}
    \item For two $\ew$-NFAs $\auta,\auta'$, using the polymorphic operation ``$\cap:\Sigma_{\ews}\times\Sigma_{\ews}\Rightarrow\Sigma_{\ews}$'', the product construction produces the $\ew$-NFA  $\auta\cap\auta'$ such that
        \[\lof{\auta\cap\auta'} =\lof\auta\cap\lof{\auta'} .\]
        Note that if $\auta,\auta'$ are NFAs then also $\auta\cap\auta'$ is an NFA.
    \item For two transducers $\trt,\trt'$, using the polymorphic operation ``$\circ:[\Sigma_{\ews},\Delta_{\ews}]\times[\Delta_{\ews},\Sigma'_{\ews}]\Rightarrow[\Sigma_{\ews},\Sigma'_{\ews}]$'', the product construction produces the transducer  $\trt\circ\trt'$  such that
        \[\rel{\trt\circ\trt'} =\rel\trt\circ\rel{\trt'}.\]
    \item For a transducer $\trt$ and an automaton $\auta$, using the polymorphic operation ``$\rinp:[\Sigma_{\ews},\Delta_{\ews}]\times\Sigma_{\ews}\Rightarrow[\Sigma_{\ews},\Delta_{\ews}]$'', the product construction produces the transducer  $\trt\rinp\auta$  such that
        \[\rel{\trt\rinp\auta} =\rel\trt\rinp\lof{\auta}.\]
        Similarly, using the polymorphic operation ``$\rout:[\Sigma_{\ews},\Delta_{\ews}]\times\Delta_{\ews}\Rightarrow[\Sigma_{\ews},\Delta_{\ews}]$'', the product construction produces the transducer  $\trt\rout\auta$  such that
        \[\rel{\trt\rout\auta} =\rel\trt\rout\lof{\auta}.\]
        These product constructions were used in \cite{Kon:2002} to answer algorithmic questions about independent languages---see Sect.~\ref{SEC:independence}.
    \end{enumerate}
\end{example}

\begin{lemma}\label{LEM:product}\prlabel{LEM:product}
	The following statements hold true about the product graph $\autg\lop\autg'=(P,C,\delta\lop\delta',I\times I',F\times F')$ of two trim labelled graphs $\autg,\autg'$ as defined in Def.~\ref{DEF:product}. 
	\begin{enumerate}
		\item $\card P=O(\card{\delta}\card{\delta'})$ and $\card{\delta\lop\delta'}\le \card{\delta}\card{\delta'}$.
		\item If the value $\beta\lop\beta'$ can be computed from the labels $\beta$ and $\beta'$ in time, and is of size, $O(\szabc\beta+\szabc{\beta'})$, then $\szw{\delta\lop\delta'}$ is of magnitude $O(\card\delta\szw{\delta'}+\card{\delta'}\szw{\delta})$ and $\delta\lop\delta'$ can be computed within time of the same order of magnitude. 
	\end{enumerate}
\end{lemma}
\begin{proof}
	As $P\sseq Q\times Q'$, Lemma~\ref{LEM:trim} implies that $\card P\le (2\card\delta+1)(2\card{\delta'}+1)$, so $\card{P}=O(\card\delta\card{\delta'})$. As we get at most one transition in $\delta\lop\delta'$ for each pair of transitions in $\delta$ and $\delta'$, we have that $\card{\delta\lop\delta'}\le \card{\delta}\card{\delta'}$. For the second statement, we have that $\delta\lop\delta'$ can be computed in time
		\[
		\sum_{(p,\beta,q)\in\delta}\>\sum_{(p',\beta',q')\in\delta'}C_{\beta,\beta'}
		\]
		where $C_{\beta,\beta'}$ is the cost of computing the value $\beta\lop\beta'$ from the labels $\beta$ and $\beta'$. Then, the statement follows using standard summation manipulations and the premise that $C_{\beta,\beta'}$ is of magnitude $O(\szabc\beta+\szabc{\beta'})$. 
\end{proof}

\begin{theorem}\label{TH:product}\prlabel{TH:product}
	If ``$\lop:B\times B'\Rightarrow C$'' is a polymorphic operation and $\autg,\autg'$ are type $B,B'$, respectively, graphs   then $\autg\lop\autg'$ is a type $C$ graph such that
	\[
	\cali(\autg\lop\autg')=\cali(\exp\autg\lop\exp\autg').
	\]
\end{theorem}
\begin{proof}
	Recall that  each transition $(p,\underline m,q)$ of $\exp\autg$ comes from a corresponding transition $(p,\beta,q)$ of $\autg$ such that $m\in\cali_1(\beta)$; and similarly each transition $(p',\underline{m'},q')$ of $\exp\autg'$ comes from a corresponding transition $(p',\beta',q')$ of $\autg'$ such that $m'\in\cali_2(\beta')$; where we used $\cali_1,\cali_2$ for the behaviours of $B,B'$. Also, if $\beta\lop\beta'\neq\undef$ and $m\lop m'\neq\undef$ then 
	\pssi $\big((p,p'),\beta\lop\beta',(q,q')\big)$ is a transition of $\autg\lop\autg'$ and
	\pnsi $\big((p,p'),\underline{m\lop m'},(q,q')\big)$ is a transition of $(\exp\autg\lop\exp\autg')$.
	\pssn
	First consider any $m\in \cI(\exp\autg\lop\exp\autg')$. Then $\exp\autg\lop\exp\autg'$  has an  accepting path 
	\[\sequ{(q_{i-1},q_{i-1}'),\underline{m_i\lop m_i'},(q_i,q_i')}_{i=1}^\ell \text{ such that } m=(m_1\lop m_1')\cdots(m_\ell\lop m_{\ell}').\]
	Then, for each $i=1,\ldots,\ell$, there is a transition $(q_{i-1},\beta_i,q_i)$ of $\autg$ with $m_i\in\cali_1(\beta_i)$; and similarly for $\autg'$, we have $m_i'\in\cali_2(\beta_i')$. Then,
	\[
	(m_i\lop m_i')\in\cali(\beta_i)\lop\cali(\beta_i')=\cI(\beta_i\lop\beta_i')
	\]
	Moreover, $\autg\lop\autg'$ has the accepting path
	\[
	\sequ{(q_{i-1},q_{i-1}'),\beta_i\lop \beta_i',(q_i,q_i')}_{i=1}^\ell
	\]
	which implies that $\cI(\beta_1\lop\beta_1')\cdots\cI(\beta_\ell\lop\beta_\ell')\sseq\cI(\autg\lop\autg')$. Hence, $m\in\cI(\autg\lop\autg')$.
	\underline{Conversely}, consider any $m\in\cI(\autg\lop\autg')$. Then $\autg\lop\autg'$  has an  accepting path 
	\[\sequ{(q_{i-1},q_{i-1}'),\beta_i\lop \beta_i',(q_i,q_i')}_{i=1}^\ell \text{ such that } m=m_1\cdots m_\ell\]
	and each $m_i\in\cI(\beta_i\lop\beta_i')=\cali_1(\beta_i)\lop\cali_2(\beta_i')$, which implies that 
	\[\mbox{each $m_i=k_i\lop k_i'$ with $k_i\in\cali_1(\beta_i)$ and $k_i'\in\cali_2(\beta_i')$.}\]
	Then, for each $i=1,\ldots,\ell$, there is a transition $(q_{i-1},\underline{k_i},q_i)$ of $\exp\autg$  and similarly there is a transition $(q_{i-1}',\underline{k_i'},q_i')$  of $\autg'$. 	Then $\exp\autg\lop\exp\autg'$ has the accepting path
	\[
	\sequ{(q_{i-1},q_{i-1}'),\underline{k_i\lop k_i'},(q_i,q_i')}_{i=1}^\ell
	\]
	which implies that $(k_1\lop k_1')\cdots(k_\ell\lop k_\ell')\in\cI(\exp\autg\lop\exp\autg')$. Hence, $m\in\cI(\exp\autg\lop\exp\autg')$. $\Box$
\end{proof}

\pnsn\textbf{How to apply the above theorem.} We can apply the theorem when we have a known product construction $\odot$ on labelled graphs $\tru,\tru'$ over monoids $M,M'$ (see Ex.~\ref{EX:product}) and we wish to apply a `higher level' version of $\odot$; that is, apply $\odot$ on labelled graphs $\autg,\autg'$ with behaviours in the monoids $M,M'$. This would avoid expanding $\autg$ and $\autg'$. We apply the theorem in Lemma~\ref{LEM:aut}.2, in Theorem~\ref{TH:compose} and in Corollary~\ref{LEM:transd:restr}.

\section{Automata and Transducers with Set Specifications}\label{SEC:sym:nfa}\prlabel{SEC:sym:nfa}
Here we present some basic algorithms on automata and transducers with set specs. These can be applied to answer the satisfaction question about independent languages (see Section~\ref{SEC:independence}).

\begin{remark}
	For every $\ew$-NFA $\auta  = ( Q, \Gamma_{\ews}, \delta, I, F )$, one can make in linear time an automaton with set specs $\auta'  = ( Q, \sspec[\Gamma], \delta', I, F )$ such that, $\delta'$ consists of all transitions $(p,\ews,q)\in\delta$ union all transitions $(p,\sone g,q)$ where $(p,g,q)\in\delta$ and $g\in\Gamma$.
\end{remark}

\begin{lemma}\label{LEM:aut}\prlabel{LEM:aut}
In the statements below, 
$$\autb=(Q, \sspec[\Gamma], \delta, I, F ) \quad\text{and}\quad \autb'=(Q', \sspec[\Gamma], \delta', I', F')$$ 
are trim automata with set specs and $w$ is a string.
\begin{enumerate}
	\item There is a $O(\sz{\autb})$ algorithm $\nonEmptyW(\autb)$ returning either a word in $\lof{\autb}$, or \None if $\lof{\autb}=\emptyset$. The decision version of this algorithm, $\emptyP(\autb)$, simply returns whether $\lof\autb$ is empty.
	\item There is a $O(\card\Gamma+\card{\delta}\szw{\delta'}+\card{\delta'}\szw{\delta})$ algorithm returning the automaton with set specs $\autb\cap\autb'$ such that $\lof{\autb\cap\autb'}=\lof{\autb}\cap\lof{\autb'}$.
	\item There is a $O(|w|\sz{\autb})$ algorithm returning whether $w\in\lof{\autb}$.
\end{enumerate}
\end{lemma}
\begin{proof}
	For the first statement, we simply use a breadth-first search (BFS) algorithm, starting from any initial state $s\in I$, which is considered visited, and stopping, either when a final state is reached (trying if necessary all initial states), or all states have been visited. In the latter case the desired algorithm returns \None (or \False). For the algorithm $\emptyP(\autb)$ nothing further is needed.  For $\nonEmptyW(\autb)$, when a non-visited state $q$ is visited from a previously visited state $p$ using a transition $e=(p,\beta,q)$, an element $x\in\lof{\beta}$ is computed in time $O(\szabc\beta)=O(\szabc e)$, using Lemma~\ref{LEM:ssets:algos}. The algorithm also constructs a string graph $G$ that will be used to find the desired word in $\lof\autb$. When the above transition is accessed and $x$ is computed then the edge $(q,x,p)$ is added to $G$. If the algorithm stops because it reached a final state $f$, then there is a unique path in $G$ from $f$ to the initial state $s$, which can be used to find the desired word in $\lof\autb$ (the path is unique as every state is visited only once). The cost of BFS is $O(|Q|+|\delta|)$, but here when an edge $e\in\delta$ is accessed the algorithm spends time $O(\szabc e)$, so the cost is 
	$$O(|Q|+\sum_{e\in\delta}\szabc e).$$
	\pnsi For the \underline{second} statement, we compute the product $\autb\cap\autb'$. As the value $\beta\cap\beta'$ of two labels can be computed in linear time, Lemma~\ref{LEM:product} implies that $\autb\cap\autb'$ can be computed in time $O(\card\Gamma+\card{\delta}\szw{\delta'}+\card{\delta'}\szw{\delta})$. Now we have
	\begin{align}
		\lof{\autb\cap\autb'} &= \lof{\exp\autb\cap\exp\autb'}\label{AL:autom:spec}\\
		  &= \lof{\exp\autb}\cap\lof{\exp\autb'}\label{AL:autom:specB}\\
		  &= \lof{\autb}\cap\lof{\autb'}
	\end{align}
	Statement \eqref{AL:autom:spec} follows from the fact that ``$\cap:\sspec[\Gamma]\times\sspec[\Gamma]\Rightarrow \sspec[\Gamma]$'' is a polymorphic operation---see Theorem~\ref{TH:product} and Ex.~\ref{EX:polymor:new}. Statement \eqref{AL:autom:specB} follows from the fact that each $\exp\autb,\exp\autb'$ is an $\ew$-NFA and the operation $\cap$ is well-defined on these objects---see Lemma~\ref{LEM:exp} and Ex.~\ref{EX:product}.
	\pnsi For the \underline{third} statement, one makes an automaton with set specs $\autb_w$ accepting $\{w\}$, then computes $\auta=\autb_w\cap\autb$, and then uses $\emptyP(\auta)$ to get the desired answer.
\end{proof}

\begin{lemma}\label{LEM:transd}\prlabel{LEM:transd}
In the statements below, $\trs=(Q, \srel[\Gamma], \delta, I, F )$ is a trim transducer with set specs and $\auta=(Q', \Gamma_{\ews}, \delta', I', F')$ is a trim $\ew$-NFA  and $(u,v)$ is a pair of words.
\begin{enumerate}
	\item There is a $O(\sz{\trs})$ algorithm $\nonEmptyW(\trs)$ returning either a word pair in $\rel{\trs}$, or \None if $\rel{\trs}=\emptyset$. The decision version of this algorithm, $\emptyP(\trs)$, simply returns whether $\rel\trs$ is empty.
	\item There is a $O(\card\Gamma+\card{\delta}\szw{\delta'}+\card{\delta'}\szw{\delta})$ algorithm returning the transducer with set specs $\trs\rinp\auta$ such that $\rel{\trs\rinp\auta}=\rel{\trs}\rinp\lof{\auta}$.
	\item There is a $O(|u||v|\sz{\trs})$ algorithm returning whether $(u,v)\in\rel{\trs}$.
\end{enumerate}
\end{lemma}
\begin{proof}
	The first statement is completely analogous to the first statement of Lemma~\ref{LEM:aut}. For the \underline{second} statement, we compute the product $\trs\rinp\auta$. As the product $\spp\rinp x$ of two labels can be computed in linear time, Lemma~\ref{LEM:product} implies that $\trs\rinp\auta$ can be computed in time $O(\card\Gamma+\card{\delta}\szw{\delta'}+\card{\delta'}\szw{\delta})$. Now we have
	\begin{align}
		\rel{\trs\rinp\auta} &= \rel{\exp\trs\rinp\exp\auta}\label{AL:transd:spec}\\
		  &= \rel{\exp\trs}\rinp\lof{\exp\auta}\label{AL:transd:specB}\\
		  &= \rel{\trs}\rinp\lof{\auta}
	\end{align}
	Statement \eqref{AL:transd:spec} follows from the fact that `$\rinp:\srel[\Gamma]\times\Gamma_{\ews}\Rightarrow \srel[\Gamma]$'' is a polymorphic operation---see Theorem~\ref{TH:product} and Ex.~\ref{EX:polymor:new}. Statement \eqref{AL:transd:specB} follows from the fact that  $\exp\trs$ is a transducer and $\exp\auta$ is an $\ew$-NFA and the operation $\rinp$ is well-defined on these objects---see Lemma~\ref{LEM:exp} and Ex.~\ref{EX:product}.
	\pnsi For the \underline{third} statement, first make two automata with set specs $\autb_u$ and $\autb_v$ accepting $\{u\}$ and $\{v\}$ respectively, then compute $\trt=\trs\rinp\auta_u\rout\auta_v$, and then use $\emptyP(\trt)$ to get the desired answer.
\end{proof}

\section{Composition of Transducers with Set Specifications} \mylabel{SEC:compose}
\pnsn
Next we are interested in defining the composition $\spp_1\circ\spp_2$ of two pairing specs in a way that $\rel{\spp_1}\circ\rel{\spp_2}$ is equal to $\rel{\spp_1\circ\spp_2}$. By Definition~\ref{DEF:spairs:rel}, the operator $\rel{}$ is defined with respect to an alphabet of reference $\Gamma$, so the value of $\spp_1\circ\spp_2$ should depend on $\Gamma$. It turns out that, for a particular subcase about the structure of $\spp_1,\spp_2$, the operation  $\spp_1\circ\spp_2$ can produce two or three pairing specs. To account for this, we define a new label set:
\pssi $\srel_+[\Gamma]$ consists of strings $\spp_1\bmoplus\cdots\bmoplus\spp_\ell$,
\pssn where $\ell\in\N$ and each $\spp_i\in\srel[\Gamma]$. Moreover we have the (fixed) label behaviour $\calr:\srel_+[\Gamma]\to2^{\Gamma^*\times\Gamma^*}$ such that
\[\rel{\spp_1\bmoplus\cdots\bmoplus\spp_\ell}= \rel{\spp_1}\cup\cdots\cup\rel{\spp_\ell}. \]

\begin{definition}\mylabel{DEF:sym:comp}
	Let $\Gamma$ be an alphabet of reference. The partial operation 
	\[\circ: \srel[\Gamma]\times \srel[\Gamma]\pto\srel_+[\Gamma] \]	
	is defined between any two pairing specs $\spp_1,\spp_2$ respecting $\Gamma$ as follows, where again $\undef$ means undefined. 
	\pssn\qquad $\spp_1\circ\spp_2=\undef$,\quad if 
	$\lof{\sprset\spp_1}\cap\lof{\spleft\spp_2}=\emptyset$.
	\pmsn Now we assume that the above condition is not true and we consider the structure of $\spp_1$ and $\spp_2$ according to Def.~\ref{DEF:spairs}\eqref{eq:pairspec} using $A,B,F,G,W,X,Y,Z$ as  set specs, where $A,B,F,G\neq\ews$---thus, we assume below that $\lof{B}\cap\lof{F}\neq\emptyset$ and $\lof{X}\cap\lof{Y}\neq\emptyset$.
	\pmsi $(W/X)\circ (Y/Z) = W/Z$
	\pmsi $(W/B)\circ (\ssame F) = W/B\cap F$
	\pssi $(W/B)\circ \big(\sdiff F G\big)  =
	\begin{cases}
		W/G,  & \text{if } |\lof{B\cap F}|\ge2\\
		W/\,G\cap\snone b,&  \text{if } \lof{B\cap F}=\{b\} \text{ and } \lof{G}\setminus\{b\}\not=\emptyset \\
		\undef, & \text{otherwise}.
	\end{cases}$
	\pmsi $(\ssame B) \circ (F/Z) \>=\> B\cap F/Z$ 
	\pmsi $(\ssame B) \circ (\ssame F) \>=\> \ssame{B\cap F}$
	

	\pmsi $(\ssame B)\circ (\sdiff F G) = \begin{cases}
 	\undef, & \text{if } \lof G=\lof{B\cap F}=\{g\} \\
 	\sdiff{B\cap F}{G}, & \text{otherwise} 
    \end{cases}$

	\pssi $(\sdiff A B)\circ (F/Z) =\begin{cases}
		A/Z, & \text{if } |\lof{B\cap F}|\ge2 \\
		A\cap\snone b/Z, & \text{if } \lof{B\cap F}=\{b\} \text{ and } \lof A\setminus\{b\}\not=\emptyset \\
		\undef& \text{otherwise.}
	\end{cases}$
	\pssi $(\sdiff A B)\circ (\ssame F) =\begin{cases}
		\undef, & \text{if } \lof A=\lof{B\cap F}=\{a\}\\		
		\sdiff{A}{B\cap F}, & \text{otherwise}
	\end{cases}$
	\pssi $(\sdiff A B) \circ (\sdiff F G)  =
	\begin{cases}
		A/G,  & \text{if } |\lof{B\cap F}|\ge3 \\
		A\cap\snone b/G\cap\snone b, & \text{if } \lof{B\cap F} = \{b\} \text{ and } \lof{A}\setminus\{b\}\\
		& \quad\quad \not=\emptyset \text{ and } \lof{G}\setminus\{b\}\not=\emptyset\\
		\bm D, & \text{if } \lof{B\cap F}=\{b_1,b_2\} 
		\\
		\undef, & \text{otherwise } 
	\end{cases}$
	\pmsn
	where $\bm D$ consists of up to three $\bmoplus$-terms as follows:
	\pssi $\bm D$ includes $A\cap\snone b_1b_2/G$, if $\lof A\setminus\{b_1,b_2\}\neq\eset$;
	\pnsi $\bm D$ includes $\sone b_1/G\cap\snone b_2$, if $b_1\in\lof A$ and $\lof G\setminus\{b_2\}\neq\eset$;
	\pnsi $\bm D$ includes $\sone b_2/G\cap\snone b_1$, if $b_2\in\lof A$ and $\lof G\setminus\{b_1\}\neq\eset$;
	\pssn and $\bm D=\undef$ if none of the above three conditions is true.
\end{definition}

\begin{remark}
	In the above definition, we have omitted cases where  $\spp_1\circ\spp_2$ is obviously undefined. For example, as $\ssame{F}$ and $\sdiff{F}{G}$ are only defined when $F,G\neq\ews$, we omit the case $(W/\ews)\circ(\ssame{F})$.
\end{remark}

\begin{remark}\mylabel{REM:sym:comp}
If we allowed $\undef$ to be a pairing spec, then the set $\srel[\Gamma]$ with the composition operation `$\circ$' would be `nearly' a semigroup: the subcase ``$(\sdiff A B) \circ (\sdiff F G)$ with $\lof{B\cap F}=\{b_1,b_2\}$'' in the above definition is the only one where the result of the composition is not necessarily a single pairing spec. For example, let the alphabet $\Gamma$ be $\{0,1,2\}$ and $A=\sone01$, $B=F=\sone12$, and $G=\sone012$. Then,
\[ \rel{\sdiff A B} \circ \rel{\sdiff F G}=\{(0,0),(0,1),(0,2),(1,0),(1,1)\}, \]
which is equal to $\rel{\{\sone0/\sone012,\,\sone1/\sone01 \}}$. This relation is not equal to $\rel{\spp}$, for any pairing spec $\spp$. 
\end{remark}

\begin{lemma}\mylabel{LEM:sym:comp}
The relation $\rel{\spp_1\circ\spp_2}$ is equal to $\rel{\spp_1}\circ\rel{\spp_2},$ for any pairing specs $\spp_1,\spp_2$ respecting $\Gamma$. 
\end{lemma}
\begin{proof} We shall use the following shorthand notation:
\pssi $Q=\rel{\spp_1}\circ\rel{\spp_2}$ and $R=\rel{\spp_1\circ\spp_2}$ and $\rel{\undef}=\eset$ . 
\pssn We shall distinguish several cases about the form of $\spp_1\circ\spp_2$ according to Def.~\ref{DEF:sym:comp}. By looking at that definition and using \eqref{EQ:leftright}, we have that
\begin{equation}\label{EQ:product}\prlabel{EQ:product}
Q,R\>\>\sseq\>\>\lof{\spleft\spp_1}\times\lof{\sprset\spp_2} 
\end{equation}
\pmsn
    \textsl{Case} $(W/X)\circ (Y/Z) = W/Z$. We have that $R=\rel{W/Z}= \lof W\times\lof Z$. As $Q$ consists of all pairs $(w,z)=(w,x)\circ(x,z)$ with $w\in\lof W, z\in\lof Z$ and $x\in\lof X\cap\lof Y$, we have that $Q=R$. 
	\pmsn 
	\textsl{Case} $(W/B)\circ (\ssame F) = W/B\cap F$. We have that $R=\rel{W/B\cap F}=\lof W\times\lof{B\cap F}$. As $Q$ consists of all pairs $(w,f)=(w,b)\circ(f,f)$ with $w\in\lof W, b\in\lof B, f\in\lof F$ and $b=f$, we have that $Q=R$.
    \pmsn
	\textsl{Case} $(W/B)\circ \big(\sdiff F G\big)$. We have three subcases. First, when $|\lof{B\cap F}|\ge2$. Then, $R=\lof{W}\times \lof{G}$.  By~\eqref{EQ:product}, $Q\sseq R$.  Now let $(w,g)\in R=\lof{W}\times\lof{G}$, and pick any $b\in \lof{B\cap F}\setminus\{g\}$. Then, $(w,g)=(w,b)\circ(b,g)\in Q$. Hence, $R\sseq Q$. 
	In the \underline{second} subcase, $\lof{B\cap F}=\{b\}$, for some $b\in\Gamma$, and $\lof G\setminus\{b\}\neq\eset$. Then, $R=\lof{W}\times\lof{G\cap\snone b}$. The claim $R=Q$ follows by noting that $Q$ consists of all pairs $(w,g)=(w,b)\circ(f,g)$ with $w\in\lof W, f\in\lof F,g\in\lof G$ and $f=b$ and $f\neq g$. In the \underline{third} subcase, $\lof G=\lof{B\cap F}=\{b\}$, so $Q=\eset$, so $Q=R$.
	\pmsn
    \textsl{Case} $(\ssame B) \circ (F/Z) \>=\> B\cap F/Z$. Analogous to case $(W/B)\circ (\ssame F) = W/B\cap F$. 
	\pmsn 
	\textsl{Case} $(\ssame B) \circ (\ssame F) \>=\> \ssame{B\cap F}$. Similar to case $(W/B)\circ (\ssame F) = W/B\cap F$, where here $R=\{(b,b)\mid b\in B\cap F\}$.
	\pmsn 
	\textsl{Case} $(\ssame B)\circ (\sdiff F G)$. First, note that $(b,g)\in Q$ iff ``$(b,g)=(b,b)\circ (f,g)$ with $g \neq f=b\in\lof{B\cap F}$''. Thus, if $\lof G=\lof{B\cap F}=\{g\}$, for some $g\in\Gamma$, then $Q=\eset=R$. Otherwise, any $(b,g)\in Q$ must be in $\rel{\sdiff{B\cap F}{G}}=R$; and conversely, if $(b,g)\in R$ then $g\neq b$ and $b\in\lof{B\cap F}$. As $(b,g)=(b,b)\circ(b,g)$ we have that $(b,g)\in\rel{\ssame{B}}\circ\rel{\sdiff{F}{G}}=Q$.
	\pmsn
	\textsl{Case} $(\sdiff A B)\circ (F/Z)$. Analogous to  case $(W/B)\circ \big(\sdiff F G\big)$.
	\pmsn
	\textsl{Case} $(\sdiff A B)\circ (\ssame F)$. First note that $(a,f)\in Q$ iff $(a,f)=(a,b)\circ (b,f)$ with $a\neq b=f\in\lof{B\cap F}$. Thus, if $\lof A=\lof{B\cap F}=\{a\}$, for some $a\in\Gamma$, then $Q=\eset=R$. Otherwise, any $(a,f)\in Q$ must be in $\rel{\sdiff{A}{B\cap F}}=R$; and conversely, if $(a,f)\in R$ then $a\neq f$ and $f\in\lof{B\cap F}$. As $(a,f)=(a,f)\circ(f,f)$ we have that $(a,f)\in\rel{\sdiff{A}{B}}\circ\rel{\ssame{F}}=Q$.
	\pmsn
	\textsl{Case} $(\sdiff A B) \circ (\sdiff F G)$. First note that, if $(a,g)\in Q$, then $(a,g)\in\lof A\times\lof G$ and
	\[(a,g)=(a,b)\circ(f,g)\text{ with } a\neq b=f\neq g, b\in\lof B\cap\lof F. \]
	 We have three subcases. First, $\lof{B\cap F}\ge3$. Then $R=\rel{A/G}$ and, by~\eqref{EQ:product}, $Q\sseq R$. Now let $(a,g)\in R$ and pick any $b\in\lof{B\cap F}\setminus\{a,g\}$. Then, $(a,g)=(a,b)\circ(b,g)\in\rel{\sdiff{A}{B}}\circ\rel{\sdiff{F}{G}}$. Hence, $R\sseq Q$. In the \underline{second} subcase, $\lof{B\cap F}=\{b\}$ and $\lof A\setminus\{b\}\neq\eset$ and $\lof G\setminus\{b\}\neq\eset$, for some $b\in\Gamma$. Then the claim $Q=R$ follows by simple inspection on the elements of $Q,R$.  In the \underline{third} subcase, $\lof{B\cap F}=\{b_1,b_2\}$, for some $b_1,b_2\in\Gamma$. The relation $Q$ can be partitioned into three subsets:
	\pssi $Q_0=\{(a,g)\mid a\in\lof A\setminus\{b_1,b_2\},g\in\lof G\}$
	\pssi $Q_1=\lof A\times\lof G\cap\{(b_1,g)\mid g\notin\lof G\setminus\{b_2\}\}$
	\pssi $Q_2=\lof A\times\lof G\cap\{(b_2,g)\mid g\notin\lof G\setminus\{b_1\}\}$
	\pssn
	Then we have that
	\pssi $\rel{A\cap\snone b_1b_2/G}=Q_0$, if $\lof A\setminus\{b_1,b_2\}\neq\eset$;
	\pnsi $\rel{\sone b_1/G\cap\snone b_2}=Q_1$, if $b_1\in\lof A$ and $\lof G\setminus\{b_2\}\neq\eset$;
	\pnsi $\rel{\sone b_2/G\cap\snone b_1}=Q_2$, if $b_2\in\lof A$ and $\lof G\setminus\{b_1\}\neq\eset$.
\end{proof}

\begin{corollary}
	The polymorphic operation ``$\circ:\srel[\Gamma]\times\srel[\Gamma]\Rightarrow \srel_+[\Gamma]$'' is well-defined by the partial operation $\circ$ in Def.~\ref{DEF:sym:comp} and the partial operation $\circ$ in Ex.~\ref{EX:monoidops}.
\end{corollary}

\begin{lemma}\mylabel{LEM:comp:complex}
	For any two pairing specs $\spp_1,\spp_2$, we have that $\spp_1\circ\spp_2$ can be computed in time $O\big(|\spp_1|+|\spp_2|\big).$
\end{lemma}
\begin{proof}
	Follows from Lemma~\ref{LEM:ssets:algos} and the fact that $|\spp_1\circ\spp_2|=O(|\spp_1|+|\spp_2|)$ as seen in Def.~\ref{DEF:sym:comp}.
\end{proof}

\begin{definition}\label{DEF:transd:compose}\prlabel{DEF:transd:compose}
	Let $\trt=(Q,\srel[\Gamma],\delta,I,F)$ and $\trs=(Q',\srel[\Gamma],\delta',I',F')$ be transducers with set specs. The transducer $\trt\tcomp\trs$ with set specs is defined as follows. First compute the transducer $\trt\circ\trs$ with labels in $\srel_+[\Gamma]$. Then, $\trt\tcomp\trs$ results when each  transition $(p,\spp_1\bmoplus\cdots\bmoplus\spp_\ell,q)$ of  $\trt\circ\trs$, with $\ell>1$, is replaced with the $\ell$ transitions $(p,\spp_i,q)$.
\end{definition}

\begin{lemma}
	We have that $\rel{\trt\tcomp\trs}=\rel{\trt\circ\trs}$.
\end{lemma}
\begin{proof}
	We show the direction $\rel{\trt\tcomp\trs}\sseq\rel{\trt\circ\trs}$; the other direction is similar. Let $(u,v)\in \rel{\trt\tcomp\trs}$. Then there is an accepting path $P=\sequ{q_{i-1},\spp_i,q_i}_{i=1}^\ell$ in $\trt\tcomp\trs$ such that 
	$$(u,v)\in\rel{\spp_1}\cdots\rel{\spp_\ell}. $$
	For each transition $e=(q_{i-1},\spp_i,q_i)$, define the triple $(q_{i-1},\spp_i',q_i)$ as follows: $\spp_i'=\spp_i$, if $e$ is in $\trt\circ\trs$; else, by Def.~\ref{DEF:transd:compose},  there is a transition $(q_{i-1},\spp_i',q_i)$ in $\trt\circ\trs$ such that $\spp_i'$ is a $\bmoplus$-sum of terms that include $\spp_i$. Then, the sequence $P'=\sequ{q_{i-1},\spp_i',q_i}_{i=1}^\ell$ is an accepting path of $\trt\circ\trs$ such that
	$$(u,v)\in\rel{\spp_1'}\cdots\rel{\spp_\ell'}. $$
	Thus, $(u,v)\in\rel{\trt\circ\trs}$.
\end{proof}

\begin{theorem}\label{TH:compose}\prlabel{TH:compose}
	For any two trim transducers $\trt=(Q,\srel[\Gamma],\delta,I,F)$ and $\trs=(Q',\srel[\Gamma],\delta',I',F')$ with set specs, $\trt\tcomp\trs$ of  can be computed in time $O(\card\Gamma+\card\delta\szw{\delta'}+\card{\delta'}\szw\delta)$. Moreover, $\rel{\trt\tcomp\trs}=\rel{\trt}\circ\rel{\trs}$.
\end{theorem}
\begin{proof}
The algorithm computes the transducer $\trt\circ\trs$ using the product construction in Def.~\ref{DEF:product}. As the composition $\spp\circ\spp'$ of any two labels of $\trt,\trs$ can be computed in linear time,  we have that  $\trt\circ\trs$ can be computed in time $O(\card\delta\szw{\delta'}+\card{\delta'}\szw\delta)$. Then, in linear time, the algorithm replaces each  transition $(p,\spp_1\bmoplus\cdots\bmoplus\spp_\ell,q)$ of  $\trt\circ\trs$, with $\ell>1$, with the $\ell$ transitions $(p,\spp_i,q)$. Now we have
	\begin{align}
		\rel{\trt\tcomp\trs} &= \rel{\trt\circ\trs}\\
		 &= \rel{\exp\trt\circ\exp\trs}\label{AL:tproductA} \\
		 &= \rel{\exp\trt}\circ\rel{\exp\trs}\label{AL:tproductB} \\
		 &= \rel{\trt}\circ\rel{\trs}.
	\end{align} 
	Statement \eqref{AL:tproductA} follows from Theorem~\ref{TH:product} and Ex.~\ref{EX:product}, and statement~\eqref{AL:tproductB} follows from Lemma~\ref{LEM:exp}.
\end{proof}

\section{Transducer Identity and Functionality}\label{SEC:functionality}\prlabel{SEC:functionality}

The question of whether a given transducer is functional is of central importance in the theory of rational relations \cite{Sak:2009}. Also important is the question of whether a given transducer $\trt$ realizes an \emdef{identity}, that is, whether $\trt(w)=\{w\}$, when $|\trt(w)|>0$. 
In \cite{AllMoh:2003}, the authors present an algorithm $\identityP(\trt)$ that works in time $O(|\delta|+|Q||\Delta|)$ and tells whether $\trt=(Q,\Sigma,\Delta,\delta,I,F)$ realizes an identity. In view of Lemma~\ref{LEM:trim}, we have that  
\begin{equation}\label{eq:identityP}
\text{for trim $\trt$, }\>\identityP(\trt) \>\text{ works in time }\> O(|\delta||\Delta|).
\end{equation}
 The algorithm $\functionalityP(\trs)$  deciding functionality of a transducer $\trt=(Q,\Gamma,\delta,I,F)$ first  constructs the \emph{square transducer} $\tru$, \cite{BeCaPrSa2003}, in which the set of transitions $\delta_{\tru}$ consists of tuples $((p,p'),y/y',(q,q'))$ such that $(p,x/y,q)$ and $(p',x/y',q')$ are any transitions in $\trt^{\ew}$. Then, it follows that $\trt$ is functional if and only if $\tru$ realizes an identity. Note that $\tru$ has $O(\card\delta^2)$ transitions and its graph size is $O(\sz{\trt}^2)$. Thus, we have that
\begin{equation}\label{eq:functionalityP}
\text{for trim $\trt$, }\>\functionalityP(\trt) \>\text{ works in time }\> O(|\delta|^2|\Delta|).
\end{equation}

\begin{lemma}\mylabel{LEM:identity}
	Let $\trs=(Q,\srel[\Gamma],\delta,I,F)$ be a trim transducer with  set specs. If any label $\spp$ of $\trs$ satisfies one of the following conditions then $\trs$ does not realize an identity. (Below, $F,G$ are set specs other than $\ews$.) 
	\pssi (C1) $\spp$ is of the form $F/G$ or $F/\ews$ or $\ews/G$, and $\card{\lof{F}}>1$ or $\card{\lof{G}}>1$.
	\pssi In the following conditions, $\spp$ is of the form $\sdiff{F}{G}$.
	\pssi (C2) $\card{\lof{F}}>2$ or $\card{\lof{G}}>2$.
	\pssi (C3) $\card{\lof{F}}=2$ and $\card{\lof{G}}=2$.
	\pssi (C4) $\card{\lof{F}}=1$ and $\card{\lof{G}}=2$ and $\lof{F}\cap\lof{G}=\eset$.
	\pssi (C5) $\card{\lof{F}}=2$ and $\card{\lof{G}}=1$ and $\lof{F}\cap\lof{G}=\eset$.
	\pssn Testing whether there is a label of $\trs$ satisfying one of the above conditions can be done in time $O(\szw{\delta})$.
\end{lemma}
\begin{proof}
	Suppose (C1) is true. We only present the subcase where $\spp=F/G$ and $\card{\lof F}>1$ (the other subcases can be dealt with similarly). Then, there are $f_1,f_2\in\lof F$, with $f_1\neq f_2$, and $y\in\lof G$. Also, $\exp\trs$ has two transitions of the form $(p,f_1/y,q)$ and $(p,f_2/y,q)$. As $\trs$ is trim, there is a path from $I$ to $p$ with some label $u/v$ and a path from $p$ to $F$ with some label $u'/v'$. As $(uf_1u',vyv'),(uf_2u',vyv')\in\rel{\exp\trs}$ and $f_1\neq f_2$, $\exp\trs$ cannot realize an identity. Now suppose one of (C2)--(C5) is true. One works as above and shows that again $\exp\trs$ cannot realize an identity.
	For the time complexity,  Lemma~\ref{LEM:ssets:algos} implies that each condition can be tested in time $O(\spp)$. For all transitions $(p,\spp,q)\in\delta$ this can be done in time $O(\szw\delta)$. 
\end{proof}

\begin{theorem}\mylabel{TH:identity}
	The question of whether a trim transducer $\trs=(Q,\srel[\Gamma],\delta,I,F)$ with set specs realizes an identity can be answered in time $O\big(\card\delta\card\Gamma\big)$.
\end{theorem}
\begin{proof}
	As $\trs$ is trim, we have that $\card{Q}\le2\card\delta+1$. First, the algorithm goes through the labels of $\trs$ and returns \False the first time a label $\spp$ satisfies one of the  conditions (C1)--(C5) in Lemma~\ref{LEM:identity}. 	Now suppose that no label $\spp$ of $\trs$ satisfies any of those conditions. Then, the algorithm  computes $\exp\trs$ and returns what $\identityP(\exp\trs)$ returns. For each transition $(p,\spp,q)\in\delta$ 	the corresponding transition(s)  $(p,x/y,q)\in\delta_{\exp}$ are computed depending on the following five \underline{cases}  about the form of $\spp$. 
	\begin{enumerate}
		\item $(\ews/\ews)$: Then, $x/y=\ews/\ews$.
		\item $(F/G)$ or $(F/\ews)$ or $(\ews/G)$: As (C1) is false, $\lof{F}=\{f\}$ and/or $\lof{G}=\{g\}$.  Then $x/y=f/g$ or $x/y=f/\ews$ or $x/y=\ews/g$, depending on whether $\spp=F/G$ or $\spp=F/\ews$ or $\spp=\ews/G$, respectively.
		\item $(\ssame{F})$: $x/y\in\{(f,f)\mid f\in\lof F\}$.
		\item $(\sdiff{F}{G})$: with $\lof{F}=\{f\}$ and $\lof{G}=\{g\}$. If $f=g$  then $\rel{\spp}=\eset$, so no label $x/y$ is defined. If $f\neq g$ then $x/y=f/g$.
		\item $(\sdiff{F}{G})$: with $\lof{F}=\{f\}$ and $\lof{G}=\{f,g\}$, or $\lof{F}=\{f,g\}$ and $\lof{G}=\{g\}$. Then $x/y=f/g$.
	\end{enumerate}	
	All cases other than the third one result in at most one transition for each $(p,\spp,q)\in\delta$. The third case results into $O(\card{\Gamma})$ transitions. Thus, $\card{\delta_{\exp}}=O(\card\delta\card\Gamma)$. Then, as $\sz{\exp\trs}=\card{\delta_{\exp}}+\card{Q}$ and $\card{Q}\le 2\card\delta+1$, we have that
	\begin{equation}\label{eq:exptrs}
			\card{\delta_{\exp}}=O(\card\delta\card\Gamma) \quad\text{and}\quad \sz{\exp\trs}=O(\card\Gamma\card\delta).
	\end{equation}
	\pnsi
	The \underline{correctness} of the algorithm follows from Lemma~\ref{LEM:identity} and the fact that $\rel{\trs}=\rel{\exp\trs}$.
	\pnsi
	Now we establish the claim about the time \underline{complexity}. The total time consists of three parts: $T_1$ = time to test conditions (C1)--(C5); $T_2$ = time to construct $\exp\trs$; and $T_3$ = time to execute $\identityP(\exp\trs)$. Lemma~\ref{LEM:identity} implies that $T_1=O(\szw\delta)$. 
	For $T_2$, we have that
	$$T_2=\sum_{e=(p,\spp,q)\in\delta}C_\spp,$$
	where $C_\spp$ is the cost of computing the set of $x/y$ for which $(p,x/y,q)\in\delta_{\exp}$. We show that $C_\spp=O(\card\Gamma)$, which implies that $T_2=O\big(\card\delta\card\Gamma\big)$. Using Lemma~\ref{LEM:ssets:algos}, testing for things like $\card{\lof F}\ge2$ can be done in time $O(\szabc{F})$ and also the same time for computing the single element of $\lof F$ when $\card{\lof F}=1$. The most time intensive task can be in the third case above: compute $\lof F$ when $F=\sone w$ and $|w|=\card{\Gamma}-1$, or $F=\snone w$ and $|w|=1$. In the former case, $\lof F$ is computed in time $O(|w|)$ by simply reading off $w$. In the latter case, we can read $\Gamma$ and make the word $u=\wo(\Gamma)$, and then use  Lemma~\ref{LEM:ssets} to compute $\sone u\cap\snone w$ in time $O(\card{\Gamma})$, which is of the form $\sone v$ and equal to $\lof F$. For $T_3$, statement \eqref{eq:identityP} implies that $\identityP(\exp\trs)$ works in time $O\big(\card{\delta_{\exp}}+\card{Q}\card\Gamma\big)$, which is $O(\card\delta\card\Gamma)$ using \eqref{eq:exptrs} and $\card{Q}\le 2\card\delta+1$. Hence, $T_3=O(\card\delta\card\Gamma)$. Thus, $T_1+T_2+T_3=O(\card\delta\card\Gamma)$ using Remark~\ref{REM:deltaszw}. $\Box$
\end{proof}

\begin{remark}\label{REM:tr:id}\prlabel{REM:tr:id}
	Consider the trim transducer $\trs$ with set specs in the above theorem. Of course one can test whether it realizes an identity by simply using $\identityP(\exp\trs)$, which would work in time $O(\card{\delta_{\exp}}\card\Gamma)$ according to \eqref{eq:identityP}. This time complexity is clearly higher than the time $O(\card\delta\card\Gamma)$ in the above theorem when $\card{\delta_{\exp}}$ is of order $\card\delta\card\Gamma$ or $\card\delta\card{\Gamma}^2$ (for example if $\trs$ involves labels $\ssame{\sany}$ or $\sany/\sany$).
\end{remark}

\begin{theorem}\mylabel{TH:funct}
	The question of whether a trim transducer $\trs=(Q,\srel[\Gamma],\delta,I,F)$ with  set specs is functional can be answered in time $O(\card\delta^2\card\Gamma)$. 
\end{theorem}
\begin{proof}
	Consider any trim transducer $\trs$ with set specs. The algorithm consists of two main parts. First, the algorithm computes $\trsinv$ and then the transducer with set specs $\tru=\trs\circ\trsinv$ using the product construction in Def.~\ref{DEF:product}. The second part is to test whether $\tru$ realizes an identity using Theorem~\ref{TH:compose}. As the composition of any two labels $\beta,\beta'$ of $\trs,\trsinv$ results in at most three labels, we have that $\tru$ has $O(\card{\delta}^2)$ transitions and is of size $O(\card\delta\szw\delta)$, and can be computed in time $O(\card\delta\szw\delta)$. Thus, testing $\tru$ for identity can be done in time $O(\card\delta^2\card\Gamma)$. So the total time of the algorithm is of order $\card\delta\szw\delta+\card\delta^2\card\Gamma$, which is $O(\card\delta^2\card\Gamma)$ by Remark~\ref{REM:deltaszw}.
	For the \underline{correctness} of the algorithm we have that 
	\begin{align}
		\rel\trs \text{ is functional } & \text{iff } \rel{\exp\trs}  \text{ is functional }\\
		& \text{iff } \rel{\exp\trs\circ(\exp\trs)^{-1}}  \text{ is an identity }\label{al:functID} \\
		& \text{iff } \rel{\exp\trs\circ(\exp\trsinv)}  \text{ is an identity }\label{al:inv} \\
		& \text{iff } \rel{\exp\trs}\circ\rel{\exp\trsinv}  \text{ is an identity }\\
		& \text{iff } \rel{\trs}\circ\rel{\trsinv}  \text{ is an identity. }
	\end{align}
	Statement \eqref{al:functID} follows from the fact that a relation $R$ is functional iff $R\circ R^{-1}$ is an identity---see also Lemma~5 of \cite{AllMoh:2003}. Statement \eqref{al:inv} follows from Lemma~\ref{LEM:expanded}.  $\Box$
\end{proof}

\begin{remark}\label{REM:tr:funct}\prlabel{REM:tr:funct}
	Consider the trim transducer $\trs$ with set specs in the above theorem. Of course one can test whether $\trs$ is functional by simply using $\functionalityP(\exp\trs)$, which would work in time $O(\card{\delta_{\exp}}^2\card\Gamma)$ according to \eqref{eq:functionalityP}. This time complexity is clearly higher than the time $O(\card\delta^2\card\Gamma)$ in the above theorem when $\card{\delta_{\exp}}$ is of order $\card\delta\card\Gamma$ or $\card\delta\card{\Gamma}^2$ (for example if $\trs$ involves labels $\ssame{\sany}$ or $\sany/\sany$).
\end{remark}

\section{Transducers and Independent Languages}\label{SEC:independence}\prlabel{SEC:independence}
Let $\trt$ be a transducer. A language $L$ is called \emdef{$\trt$-independent}, \cite{Shyr:Thierrin:relations}, if 
	\begin{equation}\mylabel{EQ:indep1}
	u,v\in L \text{ and } v\in \trt(u)\>\> \text{ implies }\> u=v.
	\end{equation}
	If the transducer $\trt$ is input-altering then, \cite{Kon:2017}, the above condition is equivalent to
	\begin{equation}\mylabel{EQ:indep2}
	\trt(L)\cap L=\emptyset.
	\end{equation}
	The \emdef{property described} by $\trt$ is the set of all $\trt$-independent languages. Main examples of such properties are code-related properties. For example, the transducer $\trt_{\rm sub2}$ describes all the 1-substitution error-detecting languages and $\trt_{\rm px}$ describes all prefix codes. The \emdef{property satisfaction} question is whether, for given transducer $\trt$ and regular language $L$, the language $L$ is $\trt$-independent. The \emdef{witness} version of this question is to compute a pair $(u,v)$ of different $L$-words (if exists) violating  condition~\eqref{EQ:indep1}.

\begin{remark}\label{REM:sat:question}\prlabel{REM:sat:question}
	The witness version of the property satisfaction question for input-altering transducers $\trs$ (see Eq.~\eqref{EQ:indep2}) can be answered in time $O(\szabc\trs\cdot\szabc\auta^2)$, where $\auta$ is the given $\ew$-NFA accepting $L$  (see \cite{Kon:2017}). This can be done using the function call
	\[\nonEmptyW(\trs\rinp\auta\rout\auta) .\]
	  Further below we show that the same question can be answered  even when $\trs$ has set specs, and this could lead to time savings. 
\end{remark}

\begin{corollary}\mylabel{LEM:transd:restr}
	Let $\trs=(Q,\srel[\Gamma],\delta,I,F)$ be a transducer with set specs and let $\autb=(Q',\Gamma_{\ews},\delta',I',F')$ be an $\ew$-NFA. Each transducer $\trs\rinp\autb$ and $\trs\rout\autb$ can be computed in time $O(\card\Gamma+ \card\delta\szw{\delta'}+\card{\delta'}\szw{\delta})$. Moreover, we have that
	\[
	\rel{\trs\rinp\autb} = \rel{\trs}\rinp\lof{\autb}\quad\text{and}\quad \rel{\trs\rout\autb}=\rel{\trs}\rout\lof{\autb}.
	\]
\end{corollary}
\begin{proof}
	The statement about the complexity follows from Lemma~\ref{LEM:product}. Then, we have
	\begin{align}
		\rel{\trs\rinp\autb} &= \rel{\exp\trs\rinp\exp\autb}\label{AL:trrestrA} \\
		&=\rel{\exp\trs}\rinp\lof{\exp\autb}\label{AL:trrestrB} \\
		&= \rel{\trs}\rinp \lof{\autb}.
	\end{align}
	Statement \eqref{AL:trrestrA} follows from Theorem~\ref{TH:product} and Ex.~\ref{EX:product}, and statement \eqref{AL:trrestrB} follows from Lemma~\ref{LEM:exp}.
\end{proof}

\begin{corollary}\label{COR:sat}\prlabel{COR:sat}
	Consider the witness version of the property satisfaction question for input-altering transducers $\trs$. The question can be answered in time $O(\szabc\trs\cdot\szabc\auta^2)$ even when the transducer $\trs$ involved has set specs. 
\end{corollary}

\begin{example}\label{EX:sat}\prlabel{EX:sat}
	We can apply the above corollary to the transducer $\trt_{\rm sub2}[\Gamma]$ of Example~\ref{EX:sym:transd}, where $\Gamma$ is the alphabet of $\autb$, so that we can decide whether a regular language is 1-substitution error-detecting in time $O(|\autb|^2)$. On the other hand, if we used the ordinary transducer $\exp \trt_{\rm sub2}[\Gamma]$ to decide the question, the required time would be $O(|\Gamma|^2\cdot|\autb|^2)$.
\end{example}
%
%

\section{Concluding Remarks}\label{SEC:conclude}\prlabel{SEC:conclude}
Regular expressions and transducers over pairing specs  allow us to describe many independence properties in a simple, alphabet invariant, way and such that these alphabet invariant objects can be processed as efficiently as their ordinary (alphabet dependent) counterparts. This is possible due to the efficiency of basic algorithms on these objects presented here. A direction for further research is to investigate how algorithms not considered here can be extended to regular expressions and transducers over pairing specs; for example, algorithms involving transducers that realize synchronous relations.
\pnsi
Algorithms on \emph{deterministic} machines with set specs might not work as efficiently as their alphabet dependent counterparts. For example the question of whether $w\in\lof{\autb}$, for given word $w$ and  DFA $\autb$ with set specs, is probably not decidable efficiently within time $O(|w|)$---see for instance the DFA with set specs in Fig.~\ref{FIG:sym:autom}. Despite this, it might be of interest to investigate this question further.
\pnsi   
Label sets can have any format as long as one provides their behaviour. For example, a label can be a string representation of a FAdo automaton, \cite{Fado}, whose behaviour of course is a regular language. At this broad level, we were able to obtain a few  results like the product construction in Theorem~\ref{TH:product}. A research direction is to investigate whether more results can be obtained at this level, or even for label sets satisfying some constraint. For example, whether partial derivatives can be defined for regular expressions involving labels other than set and pairing specs\footnote{While we have not obtained in this work the partial derivative transducer corresponding to a regular expression involving pairing specs, it is our immediate plan to do so.}.
\if\DRAFT1\textbf{\color{blue}\pnsn Anything to add?}\fi

\if\DRAFT1\pbsn\pbsn\textbf{\color{blue} Add more references ?}\fi

\bibliographystyle{plain}
\bibliography{refs}

\begin{thebibliography}{10}

\bibitem{ADKN:2005}
Parosh~Aziz Abdulla, Johann Deneux, and Lisa~KaatiMarcus Nilsson.
\newblock Minimization of non-deterministic automata with large alphabets.
\newblock In J.~Farr{\'e}, I.~Litovsky, and S.~Schmitz, editors, {\em
  Proceedings of CIAA 2005, Sydney, Australia}, volume 3845 of {\em Lecture
  Notes in Computer Science}, pages 31--42, 2006.

\bibitem{AllMoh:2003}
Cyril Allauzen and Mehryar Mohri.
\newblock Efficient algorithms for testing the twins property.
\newblock {\em Journal of Automata, Languages and Combinatorics},
  8(2):117--144, 2003.

\bibitem{antimirov96:_partial_deriv_regul_expres_finit_autom_const}
V.~M. Antimirov.
\newblock Partial derivatives of regular expressions and finite automaton
  constructions.
\newblock {\em Theoret. Comput. Sci.}, 155(2):291--319, 1996.

\bibitem{Bastos:2017aa}
Rafaela Bastos, Sabine Broda, Ant{\'o}nio Machiavelo, Nelma Moreira, and
  Rog{\'e}rio Reis.
\newblock On the average complexity of partial derivative automata for
  semi-extended expressions.
\newblock {\em Journal of Automata, Languages and Combinatorics},
  22(1--3):5--28, 2017.

\bibitem{BeCaPrSa2003}
Marie-Pierre B{\'e}al, Olivier Carton, Christophe Prieur, and Jacques
  Sakarovitch.
\newblock Squaring transducers: An efficient procedure for deciding
  functionality and sequentiality.
\newblock {\em Theoretical Computer Science}, 292(1):45--63, 2003.

\bibitem{broda11:_averag_state_compl_of_partial_deriv_autom}
Sabine Broda, Ant{\'o}nio Machiavelo, Nelma Moreira, and Rog{\'e}rio Reis.
\newblock On the average state complexity of partial derivative automata: an
  analytic combinatorics approach.
\newblock {\em International Journal of Foundations of Computer Science},
  22(7):1593--1606, 2011.
\newblock MR2865339.

\bibitem{BrzMcC:1963}
Janusz~A. Brzozowski and Edward~J. McCluskey.
\newblock Signal flow graph techniques for sequential circuit state diagrams.
\newblock {\em IEEE Trans. Electronic Computers}, 12:67--76, 1963.

\bibitem{brzozowski64:_deriv_of_regul_expres}
John Brzozowski.
\newblock Derivatives of regular expressions.
\newblock {\em J. Association for Computer Machinery}, (11):481--494, 1964.

\bibitem{caron11:_partial_deriv_of_exten_regul_expres}
Pascal Caron, Jean-Marc Champarnaud, and Ludovic Mignot.
\newblock Partial derivatives of an extended regular expression.
\newblock In Adrian~Horia Dediu, Shunsuke Inenaga, and Carlos Mart\'{\i}n-Vide,
  editors, {\em Proc. 5th LATA 2011}, volume 6638, pages 179--191. Springer,
  2011.

\bibitem{champarnaud01:_from_mirkin_prebas_to_antim}
J.~M. Champarnaud and D.~Ziadi.
\newblock From {Mirkin}'s prebases to {Antimirov}'s word partial derivatives.
\newblock {\em Fundam. Inform.}, 45(3):195--205, 2001.

\bibitem{champarnaud02:_canon_deriv_partial_deriv_and}
J.~M. Champarnaud and D.~Ziadi.
\newblock Canonical derivatives, partial derivatives and finite automaton
  constructions.
\newblock {\em Theoret. Comput. Sci.}, 289:137--163, 2002.

\bibitem{demaille16:_deriv_term_autom_multit_ration_expres}
Akim Demaille.
\newblock Derived-term automata of multitape rational expressions.
\newblock In Yo{-}Sub Han and Kai Salomaa, editors, {\em Proc. 21st {CIAA}
  2016}, volume 9705, pages 51--63. Springer, 2016.

\bibitem{Fado}
FAdo.
\newblock Tools for formal languages manipulation.
\newblock URL address: \\ \texttt{http://fado.dcc.fc.up.pt/} \hspace{2mm}
  Accessed in April, 2018.

\bibitem{Kon:2002}
Stavros Konstantinidis.
\newblock Transducers and the properties of error-detection, error-correction
  and finite-delay decodability.
\newblock {\em Journal Of Universal Computer Science}, 8:278--291, 2002.

\bibitem{Kon:2017}
Stavros Konstantinidis.
\newblock Applications of transducers in independent languages, word distances,
  codes.
\newblock In Giovanni Pighizzini and Cezar C\^ampeanu, editors, {\em
  Proceedings of DCFS 2017}, number 10316 in Lecture Notes in Computer Science,
  pages 45--62, 2017.

\bibitem{lombardy05:_deriv}
Sylvain Lombardy and Jacques Sakarovitch.
\newblock Derivatives of rational expressions with multiplicity.
\newblock {\em Theor. Comput. Sci.}, 332(1-3):141--177, 2005.

\bibitem{Man:1989}
Udi Manber.
\newblock {\em Introduction to Algorithms: A Creative Approach}.
\newblock Addison-Wesley, 1989.

\bibitem{b.g.mirkin66:_algor_for_const_base_in}
B.~G. Mirkin.
\newblock An algorithm for constructing a base in a language of regular
  expressions.
\newblock {\em Engineering Cybernetics}, 5:51---57, 1966.

\bibitem{Sak:2009}
Jacques Sakarovitch.
\newblock {\em Elements of Automata Theory}.
\newblock Cambridge University Press, Berlin, 2009.

\bibitem{Sak:2015}
Jacques Sakarovitch.
\newblock Automata and rational expressions.
\newblock {\em arXiv.org}, arXiv:1502.03573, 2015.

\bibitem{Shyr:Thierrin:relations}
H.~J. Shyr and Gabriel Thierrin.
\newblock Codes and binary relations.
\newblock In Marie~Paule Malliavin, editor, {\em S\'eminaire d'Alg\`ebre Paul
  Dubreil, Paris 1975--1976 (29\`eme Ann\'ee)}, volume 586 of {\em Lecture
  Notes in Mathematics}, pages 180--188, 1977.

\bibitem{Thom:1968}
Ken Thompson.
\newblock Regular expression search algorithm.
\newblock {\em Communications of the ACM (CACM)}, 11:419--422, 1968.

\bibitem{Vea:2013}
Margus Veanes.
\newblock Applications of symbolic finite automata.
\newblock In S.~Konstantinidis, editor, {\em Proceedings of CIAA 2013}, volume
  7982 of {\em Lecture Notes in Computer Science}, pages 16--23, 2013.

\bibitem{VHLMB:2012}
Margus Veanes, Pieter Hooimeijer, Benjamin Livshits, David Molnar, and Nikolaj
  Bjorner.
\newblock Symbolic finite state transducers: Algorithms and applications.
\newblock In John Field and Michael Hicks, editors, {\em Proceedings of the
  39th ACM SIGPLAN-SIGACT Symposium on Principles of Programming Languages,
  POPL 2012}, pages 137--150, 2012.

\bibitem{Wood:theory:of:comput}
Derick Wood.
\newblock {\em Theory of Computation}.
\newblock Harper {\&} Row, New York, 1987.

\bibitem{Yu:handbook}
Sheng Yu.
\newblock Regular languages.
\newblock In {\em \cite{FLhandbookI}}, pages 41--110.

\end{thebibliography}

\end{document}